\documentclass[10pt,a4paper]{article}

\usepackage{amsmath,amssymb,amsthm,mathrsfs}
\usepackage{stmaryrd}
\usepackage{hyperref}
\usepackage{graphicx}
\usepackage{soul}
\usepackage[thicklines]{cancel}
\usepackage{fullpage}
\usepackage[inline]{enumitem}

\usepackage{titlesec}
\usepackage{mathtools}
\usepackage{xcolor}
\usepackage{todonotes}

\theoremstyle{plain}
\newtheorem{theorem}{Theorem}[section]
\newtheorem{lemma}[theorem]{Lemma}
\newtheorem{proposition}[theorem]{Proposition}

\theoremstyle{definition}
\newtheorem{definition}[theorem]{Definition}

\newtheorem{ass}[theorem]{Assumption}

\newtheorem{remark}[theorem]{Remark}
\numberwithin{equation}{section}

\newcommand{\bE}{\mathbb{E}}
\newcommand{\E}{\bE}

\newcommand{\bR}{\mathbb{R}}
\newcommand{\bN}{\mathbb{N}}

\newcommand{\abs}[1]{\lvert#1\rvert} %

\DeclareMathOperator{\sign}{sign}

\allowdisplaybreaks
\usepackage[normalem]{ulem}
\begin{document}
    \title{A Mean-Field Game of Market Entry}
    \author{Guanxing Fu\thanks{The Hong Kong Polytechnic University, Department of Applied Mathematics and Research Centre for Quantitative Finance, Hung Hom, Kowloon, Hong Kong; guanxing.fu@polyu.edu.hk. G.
    Fu’s research is supported by Hong Kong RGC (ECS) Grant No. 25215122, NSFC Grant No. 12471453 and Grant No. 12101523, and internal grants.} \qquad  Paul P. Hager\thanks{University of Vienna, Department of Statistics and Operations Research, Kolingasse 14-16, 1090 Wien; paul.peter.hager@univie.ac.at.} \qquad Ulrich Horst \thanks{Humboldt University Berlin, Department of Mathematics and School of Business and Economics, Unter den Linden 6,
        10099 Berlin; horst@math.hu-berlin.de.
        Horst gratefully acknowledges financial support by the Deutsche Forschungsgemeinschaft through CRC/TRR 388 ``Rough Analysis, Stochastic Dynamics and Related Fields" Project B04. }}%

    \makeatletter
    \def\@maketitle{%
        \newpage
        \null
        \vskip 0.1em%
        \begin{center}%
            \let \footnote \thanks
            {\LARGE \@title \par}%
            \vskip 0em%
                {\large -- Portfolio Liquidation with Trading Constraints --\par}%
            \vskip 1.5em%
                {\large
            \lineskip .5em%
            \begin{tabular}[t]{c}%
                \@author
            \end{tabular}\par}%
            \vskip 1em%
                {\large \@date}%
        \end{center}%
        \par
        \vskip 1.5em}
    \makeatother

    \maketitle

    \begin{abstract}
        We consider both $N$-player and mean-field games of optimal portfolio liquidation in which the players are not allowed to change the direction of trading. Players with an initial short position of stocks are only allowed to buy while players with an initial long position are only allowed to sell the stock. Under suitable conditions on the model parameters we show that the games are equivalent to games of timing where the players need to determine the optimal times of market entry and exit. We identify the equilibrium entry and exit times and prove that equilibrium mean-trading rates can be characterized in terms of the solutions to a highly non-linear higher-order integral equation with endogenous terminal condition. We prove the existence of a unique solution to the integral equation from which we obtain the existence of a unique equilibrium both in the mean-field and the $N$-player game.
    \end{abstract}



    \section{Introduction}

    We consider deterministic games of optimal portfolio liquidation with finitely and infinitely many players where the players are not allowed to change the direction of trading. Players with an initial long position are only allowed to sell the stocks (``sellers''); players with an initial short position are only allowed to buy the stocks (``buyers''). This avoids any form of round-trip strategy in which a player builds up a long or short position at any point in time with the explicit goal of unwinding it later at more favorable prices. Our trading restrictions also account for the fact that compliance rules or legal restrictions may prevent brokers from changing the direction of trading when acting on behalf of clients. 

We show that our trading game is equivalent to a game of {timing} in which players determine when to optimally enter and exit the market. It turns out that in seller-dominated markets where the aggregate initial liquidity of sellers exceeds that of buyers, a seller never enters late and a buyer never exits early; the case of a buyer-dominated market is symmetric.  

{In the model equilibrium aggregate trading rates can be characterized in terms of a higher-order non-linear integral equation with \textit{endogenous} terminal condition. We  establish the existence and, under a suitable weak interaction condition, the uniqueness of a solution to this equation, and hence the existence and uniqueness of an equilibrium in our trading game.  

The integral equation depends on three key parameters: the proportion of sellers remaining in the market, the proportion of buyers who have entered the market, and the liquidity entering buyers will bring to the market in the future. The challenge is to determine the aggregate equilibrium trading rate at the terminal time, which we show to be path-dependent. It depends on the entire history of the buyers' market entries, but \textsl{not} on the history of the sellers' market exits.\footnote{This emphasizes on a mathematical level the very different roles of buyers and sellers in seller-dominated markets and highlights again the key differences between our current model and the one considered in \cite{FHH-2023} where the path-dependence is not present.}

Comparing to \cite{FHH-2023}, where a market dropout condition was imposed for early liquidation, we find that in a seller-dominated market the ``no change in trading direction'' constraint imposed in the current paper has the same effect on the seller side, in the sense that a seller starts trading immediately and leaves the market as soon as her portfolio process hits zero.
As a result, the best responses of the sellers to the aggregate market take the same form as under the dropout condition.

The situation is very different for the buyers. In \cite{FHH-2023} buyers with small initial portfolios initially sell the asset in equilibrium and buy it back later to benefit from favorable price trends that outweigh liquidity costs. This form of round-trip strategy is not allowed in our current model. As a consequence, some buyers will postpone market entry, which results in dynamic market entries and the said path-dependence of the equilibrium equation, which makes the equilibrium analysis much more challenging.  

    \subsection{Portfolio liquidation models}

    Models of optimal portfolio liquidation have received substantial consideration in the financial mathematics literature in recent years. Starting with the work of Almgren and Chriss \cite{AC-2001} existence and uniqueness of solutions to single-player problems in different settings have been established by a variety of authors including \cite{AJK-2014,bank:voss:16,FruthSchoenebornUrusov14, FHX2, GatheralSchied11, GH-2017, HXZ-2020,Kratz14, KP-2016, OW-2013,PZ-2018}. One of the main characteristics of portfolio liquidation models is a singular terminal condition of the value function induced by the liquidation constraint. The singularity causes substantial technical difficulties when solving the value function and/or applying verification arguments.

\subsubsection{Liquidation games}

    Mean-field liquidation games with market impact but without trading constraints and without strict liquidation constraints have been analyzed by many authors. Cardaliaguet and Lehalle \cite{CL-2018} considered a {mean-field game} (MFG) where each player has a different risk aversion. Casgrain and Jaimungal \cite{C-Jai-2018,C-Jai-2018b} considered games with partial information and different beliefs, respectively. Huang et al. \cite{HJN-2015} considered a game between a major agent that is liquidating a large number of shares and many minor agents that trade against the major player.

    Finite-player market impact games with and without strict liquidation constraint and transient market impact were studied in, e.g.~\cite{Schied-2020, Schied-2017b,Schied-2019,Strehle-2018} and more recently by Micheli et al \cite{Micheli-2023}  and Neumann and Vo\ss ~\cite{NeumannVoss-2023}; games with permanent impact were studied in, e.g. \cite{Carlin2007,Drapeau2019,FH-2020}. Mean-field liquidation games {with} strict liquidation constraint have been analyzed in \cite{FGHP-2018, FHX1}.

    A mean-field liquidation game with permanent impact and market dropout has recently been considered in our accompanying paper \cite{FHH-2023}. Under the dropout condition, a player exits the market as soon as her portfolio process hits zero. The condition avoids round-trips where players with zero initial position trade the asset to benefit from favorable future market dynamics. Beneficial round-trips are usually regarded as a form of statistical arbitrage and should thus be avoided. The dropout condition also avoids the ``hot potato effects'' that occur in \cite{Schied-2017b, Schied-2019} where different players repeatedly take long and short positions in the same asset. However, it does not prevent players from changing the direction of trading.

    In models with only sellers or only buyers, the dropout condition turns out to be equivalent to the ``no change in trading direction'' condition. However, when buyers and sellers interact in the same market, an example in \cite{FHH-2023} shows that the dropout condition does not prevent players from changing the direction of trading. In markets dominated by sellers (buyers), a weak form of round-trip strategies, where buyers (sellers) with small initial conditions may take advantage of price trends and benefit from first selling (buying) the asset and then buying (selling) it back at better prices, may still emerge. Our ``no change in trading condition'' is much stronger than the dropout condition in \cite{FHH-2023} and avoids any form of round-trip strategies.

\subsubsection{Mean-field games of timing}

We prove that our game is equivalent to a game of timing, 
   in which the players need to determine the optimal times of market entry {\sl and} exit. In particular, in seller-dominated markets, the equilibrium dynamics will depend on the entire history of market entries (market exits in buyer-dominated markets). This is in sharp contrast to earlier work like \cite{FHH-2023,PUR} where only exit times had to be determined and the equilibrium dynamics only depended on the overall proportion of early exits. 

    The literature on (mean-field) games of optimal entry and exit is still sparse, especially when both entry and exit times need to be determined. The paper that is conceptually closest to ours is the one by A\"id et al \cite{DuTa-2021b}. They consider an MFG of electricity production where energy producers using conventional, respectively, renewable resources need to decide when to exit, respectively enter the market.  In our model, the players need to determine \textsl{both} entry and exit times.

    Dumitrescu et al. \cite{DuTa-2021} and Bouveret et al \cite{DuTa-2020} develop relaxed solutions approaches to solve MFGs where the representative agent chooses both the optimal control and the optimal time to exit the game. Campi and coauthors \cite{campi2021, campi2018, campi2020} consider special classes of MFGs with drop-out (exit). Even if not explicitly formulated as stopping problems, drop-out conditions implicitly involve a choice of optimal exit times. Carmona et al \cite{CDL-2017}, and Nutz \cite{N-2018} use probabilistic methods to solve MFGs arising in models of bank runs that can also be viewed as MFGs of market drop-out. No entry times are to be determined in these models. We shall see that in our setting determining equilibrium entry and exit times requires very different approaches.

    \subsection{Solving our liquidation games}

    We solve our MFG and $N$-player game of optimal liquidation within a common mathematical framework. The key mathematical challenge is to solve the multi-dimensional (in the $N$-player game), highly non-linear system of forward-backward equations that characterizes the candidate equilibrium trading strategies.  %

    \subsubsection{Decoupling and candidate entry/exit times}
    
   To overcome the said problem, we follow \cite{FHH-2023} and first decouple the multi-dimensional forward-backward system by replacing the average trading rate in the co-state dynamics by an exogenous trading rate. This is precisely what one would do in MFGs; in our case, this approach also works for $N$-player games. The decoupling reduces the multi-dimensional system to a family of independent systems, but the new systems are still non-linear. In a second step, we introduce a family of auxiliary linear systems on different time intervals in terms of which we characterize candidate optimal market entry and exit times of a representative buyer and seller. It turns out that in seller-dominated markets the candidate entry times for seller and the candidate exit times for buyers are trivial (vice versa in buyer dominated markets). In particular, we  only need to determine {\sl either} the equilibrium exit {\sl or} entry times.\footnote{We emphasize, that this is an equilibrium property; a priori both times need to be determined.} %
    
     Our candidate exit times for seller-dominated markets are the same as in \cite{FHH-2023}. As a result, we only need to determine the buyers' entry times. Entry times require a different approach. Optimal exit times turn out to be the first times where the portfolio process is zero; by contrast, entry times turn out to be the first times where the derivative of the portfolio process is different from zero.

    We prove that only buyers with comparably small positions enter a seller-dominated market late. This result is very intuitive. In a model with trading constraints, players with small enough position could potentially benefit from favorable price trends that outweigh the additional impact cost a player incurs when she initially increases a position that she actually needs to unwind. Under our trading constraints, these are precisely the players that enter late (respectively exit early in buyer-dominated markets).

    With the candidate entry and exit times in hand, we derive candidate best response strategies for buyers and sellers to {\sl exogenous} aggregate trading rates, in terms of the solutions to {\sl unrestricted} trading problems on the resulting endogenous trading intervals. It turns out that the corresponding portfolio processes are strictly monotone, hence admissible and optimal even under the ``no change of trading direction'' condition.

    In terms of the candidate best response functions, we derive a general fixed-point equation for the candidate equilibrium mean trading rate and prove that the fixed-point equation can be rewritten in terms of a higher-order, non-linear, path-dependent integral equation with endogenous terminal condition. 

\subsubsection{Equilibrium dynamics}
    
    Although the approach of reducing our game to a game of timing and the idea of rewriting the fixed-point equation as a non-linear integral equation follows \cite{FHH-2023}, the resulting equation is very different and the analysis of the equilibrium dynamics is much more challenging.   
    
    Compared to the market dropout situation studied in \cite{FHH-2023}  the continuous influx of players adds additional nonlinear components to the fixed-point equation. Moreover, the terminal condition of the equilibrium equation is much harder to identify, and it turns out that it depends on the entire history of market entries. As a result, our equilibrium equation is fully path-dependent. This is {\sl not} the case in \cite{FHH-2023} where all players enter the market at the initial time. 

    Our key observation is that the path-dependence
can be disentangled by introducing two free parameters that represent the aggregate trading rate at the terminal time and - loosely speaking - the largest initial endowment that a seller that exits early can hold. We show that solving the fixed-point equation is equivalent to solving a {\sl two-dimensional} root-finding problem that incorporates parameter-dependent solution to a still nonlinear and still higher-order, yet non path-dependent, integral equation. The corresponding root-finding problem in \cite{FHH-2023} reduces to {\sl one-dimensional} problem, as the second free parameter is a function of the first. 
    
    The main difficulty when solving the two-dimensional root-finding problem is to verify monotonicity properties of the solution map with respect to our parameters, which we achieve by identifying Volterra integral equations for the corresponding partial derivatives and applying a suitable comparison principle. %
    The Volterra kernel depends non-linearly on the solution, is generally discontinuous, and exhibits a singularity at the origin; as a result, it does not fall into any standard class of kernels.
    
    We prove the non-negativity of the kernel function and establish growth bounds that carry over to the partial derivatives of the solution mapping.  This allows us to prove that the root-finding problem has a solution and that the solution is unique under a bound on the impact of buyers or sellers on the market dynamics, depending on which side holds the smaller initial position. Moderate influence conditions are standard in the game theory literature when proving uniqueness of Nash equilibria. In various economic settings they have, for instance, been imposed in, e.g. ~\cite{H-2005, HS-2006}. In market impact games weak interaction conditions have been imposed in, e.g.~\cite{FGHP-2018, FHX1, Micheli-2023}.
    
    Our theoretical analysis is accompanied by extensive numerical simulations. Our simulations suggest that convergence to the MFG equilibrium is fast and that the MFG provides a good approximation for games with 15 players or more. Our simulations also suggest that trading constraints may lower aggregate costs in markets with strong permanent impact. This result is very intuitive. Without constraints buyers may choose to initially sell additional assets in seller dominated markets, thereby amplifying a downward price trend that results in additional trading costs for the majority of market participants.

    The remainder of this paper is organized as follows. In Section \ref{sec:model} we introduce our liquidation games. Section \ref{Sec3-new} studies the auxiliary forward-backward systems and derives the candidate best-response function for buyers and sellers separately. The equilibrium analysis is carried out in Section \ref{sec:eqana}. Section \ref{sec:examples} illustrates the impact of our trading constraint on equilibrium trading. Section \ref{sec:conclusion} concludes.

\medskip

\section{Model and main results}\label{sec:model}

In this section, we introduce the portfolio liquidation game with trading constraint and state our main results, including an existence and uniqueness of equilibrium result under a suitable bound on the player's risk aversion. In particular, we show that in seller dominated markets where the majority of the liquidity is on the sell side, in equilibrium buyers with small initial positions enter the market late, and sellers with small initial positions leave the market early; in buyer dominated markets buyers with small initial portfolios leave the market early and sellers with small initial positions enter late.

\subsection{The trading game}

We first consider a liquidation game among $N$ players in which player $i \in \{1, ..., N\}$ holds an initial portfolio of $x_i \in \bR$ shares that he or she needs to close over the time interval $[0,T]$. If the initial position is positive, the player needs to sell the stock; else, he or she needs to buy it. The distribution of the players' initial portfolios is denoted by
    \[
        \nu^{N}(dx) := \frac{1}{N}\sum_{i=1}^N \delta_{x_i}(dx).
    \]

    Following the majority of the liquidation literature we assume that only absolutely continuous trading strategies are allowed. The portfolio process of player $i$ is hence given by
    \[
        X^i_t = x_i-\int_0^t \xi^i_s \,ds, \quad t \in [0,T]
    \]
  subject to the binding liquidation constraint
    \[
    	X^i_T = 0,
    \] 
    where $\xi^i_t$ denotes the trading rate at time $t \in [0,T]$; positive rates indicate that the player is selling the asset; negative rates indicate that he or she is buying it.  

    We assume that the unaffected price process, against which the trading costs are benchmarked, follows some Brownian martingale $S$ and that the transaction price process of player $i$ is of the form
    \[
        \tilde S^i_t = S_t - \int_0^t \kappa \overline\xi^N_s ds - {\color{black} \frac 1 2}\eta_t \xi^i_t, \quad t \in [0,T]
    \]
    for some deterministic positive market impact process  $\eta$ and constant $\kappa$, and
    \begin{align*}
        \overline\xi^N_t ~:=~ \frac{1}{N}\sum_{i=1}^{N} {\color{black}\xi^i_t}
    \end{align*}
    denotes the average trading rate throughout the entire universe of players. That the permanent impact factor $\kappa$ and the instantaneous impact factor $\eta$ is the same for all players accounts for the fact that all players are trading in the same market.

    The assumption that permanent market impact depends on aggregate behavior is standard in the literature on liquidation games, see e.g. \cite{CL-2018,CL-2015,C-Jai-2018b,FGHP-2018, FHH-2023}. By contrast, the instantaneous impact depends on individual, not aggregate demand. As different traders never consume liquidity at exactly the same time in practice, it is reasonable to assume that instantaneous impact always only affects one player.

    The player's liquidation cost $C^i$ is defined as the difference between the book value and the proceeds from trading:
    \[
        C^i = x_i S_0 - \int_0^T \tilde S^i_t \xi^i_t\, dt.
    \]
    Doing integration by parts and taking expectations the martingale terms drops out and the expected liquidation cost equals
    \[
        \mathbb{E}[C^i] = \int_0^T \left( \frac{1}{2}\eta_t {\color{black}(\xi^i_t)^2+\kappa \bar \xi^N_t X^i_t} \right) \, dt.
    \]

    Introducing an additional risk term  $\frac{1}{2}\lambda_t (X^i_t)^2$ for some deterministic non-negative function $\lambda$ that penalizes slow liquidation, the cost functional for a generic player $i$ given the vector {\color{black} $\xi^{-i}=(\xi^j)_{j \neq i}$} of all the other players' strategies equals
    \[
        J(\xi^i; \xi^{-i}):=\int_0^T \left(\frac{1}{2}\eta_t(\xi^i_t)^2+\frac{\kappa X^i_t}{ N}  \sum_{j=1}^N \xi^j_t +\frac{1}{2}\lambda_t (X^i_t)^2\right)\,dt.
    \]

    The above cost function is standard in the liquidation literature. Departing from the standard literature,  we assume that the players are not allowed to change the direction of trading. The set of admissible trading strategies of player $i$ is hence given by the set
    \begin{equation*}
        \mathcal A_{x_i}:=\left\{\xi^i \in L^2([0,T]) \;\bigg|\;  \sign(x_i) \xi^i_t \geq 0 ~ \mbox{and} ~ X^i_T = 0 \right\}
    \end{equation*}
    of all square integrable strategies that satisfy the trading and the liquidation constraint, and her optimization problem reads
    \begin{equation}\label{eq:N-player}
        \min_{\xi^i \in {\mathcal{A}}_{x_i}} J(\xi^i;\xi^{-i}) \quad \mbox{ s.t. } \quad dX^i_t = - \xi^i_t dt,\quad X_0^i=x_i.
    \end{equation}
    An admissible strategy profile $\xi^* = \big(\xi^{*,1}, ... , \xi^{*,N} \big)$ is a Nash equilibrium if for all $\xi^i \in \mathcal{A}_{x_i}$ and all $i=1, ..., N$,
    \[
        J(\xi^{*,i};\xi^{*,-i}) \leq J(\xi^{i};\xi^{*,-i}).
    \]

    In the corresponding MFG the distribution of the players' initial positions is denoted by $\nu$, the average trading rate is replaced by an exogenous trading rate $\mu$, the representative player's cost functional is given by
    \[
        J(\xi; \mu):=\int_0^T \left(\frac{1}{2}\eta_t\xi^2_t+\kappa \mu_t X_t+\frac{1}{2}\lambda_t X^2_t\right)\,dt
    \]
    and her control problem reads
    \begin{equation}
        \label{opt-rp}
        \min_{\xi \in \mathcal A_x} J(\xi;\mu) \quad \mbox{ s.t. } \quad dX_t = - \xi_t dt,\quad X_0 \sim \nu.
    \end{equation}
    Given an initial distribution $\nu$ of portfolios and the optimal trading rates $\xi^{*,x,\mu}$ for the representative player with initial position $x$ as a function of the exogenous mean trading rate $\mu$, the equilibrium condition both in the MFG and the finite-player game reads
    \[
        \mu = \int_{\mathbb{R}}\xi^{\ast, x,\mu} \nu(dx).
    \]

\subsection{Assumptions and main result}

In this section, we state our existence and uniqueness of equilibrium result. We shall distinguish between buyer- and seller-dominated markets.  

\begin{definition}
	We say that a market is {\sl seller-dominated} if the initial aggregate position of sellers exceed that of the buyers, i.e.~if  $\int x \nu(dx)>0$. We call the model {\sl buyer-dominated} if the buyers' aggregate initial positions exceed that of the sellers. 
\end{definition}

    We proceed under the following standing assumption on the model parameters, which is binding throughout the paper. The fact that the permanent impact factor $\kappa$ is assumed to be constant is needed to unify the verification arguments for the MFG and the $N$-player game. If only the MFG is considered, then $\kappa$ can be chosen to be a continuously differentiable function of time.

    \begin{ass}
        \label{ass:single-player} The cost coefficients satisfy $$ \lambda \in L^{\infty}([0,T];[0,\infty)), \quad 1/\eta, \eta\in C^{1}([0,T];(0,\infty)), \quad \mbox{and} \quad \kappa > 0.$$
    \end{ass}

    We are now ready to state the main result of this paper. Its proof is given in the following sections.
    To unify notation, in the following, we denote by $\nu$ a generic distribution of initial positions, which reduces to the empirical distribution of $\{x_i\}_{i=1, \dots, N}$ in the case of an $N$-player game.

    \begin{theorem}\label{thm:main}
        Suppose that the distribution of the players' initial portfolios has a finite absolute first moment, i.e. 
        \[
        		\int |x| \nu(dx) < \infty
        \]
        and that the instantaneous impact parameter and the risk aversion coefficient satisfy at least one of the following conditions:
        \begin{itemize}
            \item The function $\lambda \leq \lambda_0$, where $\lambda_0$ is an explicit positive constant\footnote{See Appendix \ref{app:A} for the derivation of the explicit expression $\lambda_0= 2 (T^2\Vert \eta^{-1}\Vert_{\infty}^2 \|\eta\|_\infty)^{-1} e^{- 2T \kappa \Vert \eta^{-1}\Vert_{\infty}}$.} that depends only on the time horizon $T$, the permanent impact factor $\kappa$ and the process $\eta$ (e.g. $\lambda = 0$).
            \item The product $\lambda\eta$ is non-decreasing (e.g. constant parameters.)
        \end{itemize}
        Then the following holds:
        \begin{itemize}
            \item[(i)] Both the $N$-player and the MFG admit a Nash equilibrium such that the aggregate equilibrium trading rate does not change its sign and such that
            \begin{itemize}
            	\item the mapping $t \mapsto \eta_t \mu_t$ is non-increasing if the market is seller-dominated
		\item the mapping $t \mapsto \eta_t \mu_t$ is non-decreasing if the market is buyer-dominated.
        \end{itemize}
        {In this equilibrium, if the market is seller- (buyer-)dominated, then  sellers (buyers) never enter late while buyers (sellers) never exit early.}  
            \item[(ii)] %
            There exists a constant\footnote{An explicit form of $\nu_0$ is difficult to obtain, as it depends on an a priori bound of an intricate Volterra integral equation (see the proof of Lemma~\ref{lem:differentiation_bc} and Theorem \ref{thm:existence2}).} $\nu_0 \in [0,1)$ that depends only $T$ and the model parameters $\lambda, \eta$ and $\kappa$ such that if 
            \begin{itemize}
                \item $\nu\big((-\infty,0]\big) < \nu_0$ for the case of a seller-dominated market,
                \item $\nu\big([0,\infty)\big) < \nu_0$ for the case of a buyer-dominated market,
            \end{itemize}
            that is, if the non-dominating side of the market is small enough,
            then the equilibrium from (i) is  unique with continuous  aggregate equilibrium trading rate satisfying the above monotonicity properties.

            \item[(iii)] Assume that the initial distribution $\nu$ of the MFG has a density and satisfies the uniqueness condition in (ii). Let $(\nu^N)_{N \in \mathbb{N}}$ be a sequence of initial distributions for the finite-player game that converges weakly to $\nu$ and such that\footnote{The convergence assumptions are equivalent to convergence in the Wasserstein $W_1$ distance.}
            $$\lim_{N\to\infty}\int_{\bR} |x|\nu^N(dx) = \int_{\bR} |x|\nu(dx).$$
            Then the corresponding finite-player aggregate equilibrium trading rates $(\mu^N)_{N\in\mathbb N}$  converge to the mean-field equilibrium $\mu$:
            \[
                \lim_{N\rightarrow\infty}\sup_{0\leq t\leq T}|\mu^N_t-\mu_t|=0.
            \]
            \item[(iv)] If the market is initially equilibrated, i.e.~if $\int x \nu(dx)=0$, then $\mu \equiv 0$ is an equilibrium.
        \end{itemize}
    \end{theorem}

It follows from the preceding theorem that a continuous, monotone equilibrium aggregate trading rate exists whenever the model parameters are constant. In this case, the equilibrium is unique in this class whenever the impact of buyers (sellers) in a seller-dominated (buyer-dominated) market is small enough.

\subsection{Characterization of equilibria}

   Our equilibrium trading rates can be charaterized in terms of high-dimensional non-linear forward-backward system of ordinary differential equations. In fact, given the trading rates $\xi^{-i}{=(\xi^j)_{j\neq i}}$ of all other players the Hamiltonian associated with the optimization problem of player $i$ is given by
    \[
        H(t,\xi^i,X^i,Y^i; \xi^{-i} )=-\xi^i Y^i + \frac{1}{2}\eta_t (\xi^i)^2+\kappa \bar \xi^N_t  X^i +\frac{1}{2}\lambda_t (X^i)^2.
    \]
    In the corresponding MFG the average rate $\bar \xi^N$ is to be replaced by a generic trading rate $\mu$. 
    
    {\color{black}Setting $y_{\sign(x)}:=y\vee 0$ if $x>0$ and $y_{\sign(x)}:=y\wedge 0$ if $x\leq 0$ for $x,y \in \mathbb R$},  minimizing the Hamiltonian pointwise and taking the trading constraint into consideration yields the candidate equilibrium strategies
    \begin{equation}
        \label{xi*8}
        \xi^{i}_t :=
        \left(\frac{Y^i_t  - {\frac 1 N \kappa X^i_t}   }{\eta_t}\right)_{\sign(x_i)}
    \end{equation}
    in terms of a solution to the coupled, non-linear forward-backward differential equation
    \begin{equation}
        \label{eq:XY-i2}
        \left\{\begin{split}
                   \bigg.  \dot X^i_t=&~-\left( \frac{Y^i_t  - {\frac 1 N \kappa X^i_t} }{\eta_t} \right)_{\sign(x_i)}  \\
                   \bigg.  -\dot Y^i_t=&~(\lambda_t X^i_t+\kappa \bar \xi^N_t  ) \\
                   X^i_0=&~x_i,\quad X^i_T = 0
        \end{split}\right., \qquad \text{for a.e. } t\in[0,T].
    \end{equation}

\begin{remark}
We notice that the terminal state of the adjoint equation is unknown, due to the liquidation constraint on the state process. The terminal condition needs to be determined in equilibrium.
\end{remark}

{\color{black}Solving the high-dimensional, non-linear forward-backward ODE system \eqref{eq:XY-i2} is challenging. To solve the system we proceed in two main steps. 

\begin{itemize}
	\item In Section 3 we decouple the system by replacing the aggregate trading rate in the co-state dynamics by an exogenous process $\mu$. The decoupled system allows us to obtain best responses $\xi^{i,\mu}$ to an exogenous, respectively equilibrium trading rates $\mu$, for MFG, respectively, the $N$-player game. 
	\item In Section 4 we solve the fixed-point problem of finding an exogenous trading rate $\mu^*$ that coincides with the average trading rate $\overline \xi^{\mu^*}$ associated with the strategies $\xi^{i,\mu^*}$, and prove that any such $\mu^*$ forms an equilibrium aggregate trading rate.     
\end{itemize}
}
We proceed under the standing assumption of a seller dominated market. The case of a buyer-dominated market is symmetric.

\section{Conditional best responses}\label{Sec3-new}

Following the approach introduced in \cite{FHH-2023} we now reduce the coupled system \eqref{eq:XY-i2} to a single forward-backward system by replacing the aggregate trading rate $\bar \xi^N$ by an exogenous process $\mu$. This is precisely what one would do in the MFG; in our model this approach also works for finite player games. 

{\color{black}
While our decoupling approach reduces the high-dimensional, coupled system to a family of independent systems, the new system is still non-linear. In a second step we, therefore, introduce a family of auxiliary linear systems on arbitrary time intervals $[\sigma,\tau] \subseteq [0,T]$ that can be solved in closed form. 

The closed-form solutions yield a family of auxiliary strategies $\xi^{x,\mu,\sigma,\tau}$ depending on a player's initial endowment $x$, in terms of which we shall obtain candidate optimal market entry and exit times $\sigma^*$ and $\tau^*$. We show that in a seller-dominated markets, sellers never enter late while buyers never exit early. In particular, for buyers we only need to determine $\sigma^*$ and for sellers we only need to determine $\tau^*$. 

In a final step we prove that the resulting strategy $\xi^{x,\mu,\sigma^*,\tau^*}$ is an optimal trading strategy in a model with exogenous impact function $\mu$ for MFG. For the $N$-player game, the best-response property will only hold in equilibrium. 
}

\subsection{Decoupling the forward-backward system}
    
Motivated by the analysis of MFGs we now consider - for any $\delta \in [0,1]$, any initial position $x \in \bR$ and any exogenous trading rate $\mu$ - the auxiliary forward-backward system
    \begin{equation}
        \label{eq:XY-i}
        \left\{\begin{split}
                   \bigg.  \dot X_t=&~-\left( \frac{Y_t  - {\delta \kappa X_t} }{\eta_t} \right)_{\sign(x)}  \\
                   \bigg.  -\dot Y_t=&~(\lambda_t X_t+\kappa \mu_t  ) \\
                   X_0=&~x,\quad X_T = 0
        \end{split}\right., \qquad \text{for a.e. } t\in[0,T].
    \end{equation}

    The case $\delta=0$ corresponds to the MFG. In this case the above system describes the forward-backward system associated with the representative player's optimization problem, and for any given exogenous trading rate $\mu$ we expect a solution $(X^\mu,Y^\mu)$ to yield the representative agent's best response
    \begin{equation}
        \xi^\mu := \left( \frac{Y^\mu - \delta \kappa X^\mu}{\eta}\right)_{\sign(x)}.
    \end{equation}

    The case $\delta=\frac{1}{N}$ corresponds to the $N$-player game. In this case we expect  the best response property to hold in equilibrium. 

 We proceed under the assumption that the exogenous trading rate $\mu$ is strictly positive. For seller-dominated markets this condition will be verified in equilibrium. For technical reasons we also need to assume that the map $t \mapsto \eta_t \mu_t$ is non-increasing. This assumption, too, will be verified in equilibrium.

    \begin{ass}
        \label{ass:mu_sign}
        \begin{itemize} 
            \item[(i)] The function $\mu: [0,T] \to \bR$ does not change sign and w.l.o.g. $\mu > 0$.
            \item[(ii)] The function $t \mapsto \eta_t \mu_t$ is non-increasing (non-decreasing if $\mu < 0$).
        \end{itemize}
    \end{ass}

    \subsection{Reduction to a game of timing}

We are now going to reduce the non-linear system \eqref{eq:XY-i} to a family of linear ones, from which we shall deduce that our trading game can be reduced to a game of timing, where the players need to determine optimal market entry and exit times. To this end, we consider, for any pair $0 \leq \sigma < \tau \le T$ the linear ODE system
    \begin{equation}
        \label{eq:FBODE} %
        \left\{\begin{split}
                   \bigg. \dot X_t=&~- \frac{Y_t  - {\delta\kappa X_t}     }{\eta_t}  1_{\{\sigma \leq t \le \tau\}}\\
                   \bigg.  -\dot Y_t=&~    \lambda_t X_t+\kappa\mu_t \\ %
                   X_\sigma=&~x,\quad  X_\tau = 0
        \end{split}\right., \qquad \text{for a.e. } t\in[0,T] .
    \end{equation}
    
    Our goal is to identify entry and exit times $\sigma$ and $\tau$ such that the solutions to the constrained system \eqref{eq:XY-i} and the unconstrained system \eqref{eq:FBODE} coincide. 
    
    The unconstrained system can be solved in closed form. To see this, we denote by
    $(A^{{\delta}}, B^{\delta,\tau})$ the unique solution to the following singular Riccati equation on $[0,T]$:
    \begin{equation}
        \label{eq:AB+}
        \left\{\begin{split}
                   -\dot A_t=&~-\frac{A^2_t}{\eta_t} {+\delta\frac{\kappa}{\eta_t}A_t} +\lambda_t\\
                   -\dot B_t=&~\left(-\frac{A_tB_t}{\eta_t}+\kappa\mu_t\right) 1_{\{t \leq\tau\}}\\
                   \lim_{t\nearrow T}A_t=&~\infty,\quad B_{\tau} =0
        \end{split}\right.
    \end{equation}

    {\color{black} The analysis in \cite[Section 2.3.2]{FHH-2023} }shows that, for any exit time $\tau \in (0,T]$, solving the Riccati equation on the interval $[\sigma,\tau]$ is equivalent to solving the ODE system \eqref{eq:FBODE}, and the explicit solution is given by
    \begin{equation}
        \label{eq:explicit-X}
        \left\{   \begin{split}
                      X^{\delta,\sigma,\tau}_t &=   xe^{-\int_\sigma^t \frac{A^{{\delta}}_r  {-\delta\kappa} }{\eta_r}\,dr}    -\int_\sigma^t \frac{1}{\eta_s} e^{-\int_s^t \frac{A^{{\delta}}_r   {-\delta\kappa}    }{\eta_r}\,dr }\int_s^{\tau} \kappa \mu_u e^{-\int_s^u\frac{A^{{\delta}}_r}{\eta_r}\,dr}\,du\,ds \\
                      Y^{\delta, \sigma,\tau}_t &= A^\delta_{t} X^{\delta, \sigma,\tau}_t + B^{\delta,\tau}_t.
        \end{split} \right.
    \end{equation}

\begin{remark}
    We emphasize that the Riccati equation \eqref{eq:AB+} can be solved for any $\delta \in [0,1]$ and any pair $0 \leq \sigma < \tau \leq T$, and hence that the process $\Big( X^{\delta,\sigma,\tau}, Y^{\delta,\sigma,\tau} \big)$ is well defined for any such triple. However, in general we cannot expect the process $X^{\delta,\sigma,\tau}$ to satisfy the liquidation constraint. Hence solving the systems \eqref{eq:FBODE} and \eqref{eq:AB+} is not equivalent in general; this is true only if we know \textsl{a priori} that $\tau$ is an exit time, i.e., that\footnote{For the process $X^{\delta,\sigma,\tau}$ to satisfy the liquidation constraint for any given $\tau$, one has to replace the singular terminal condition in \eqref{eq:AB+} by $\lim_{t\nearrow \tau}A_t=~\infty$, in which case the process $A$ would depend on $\tau$.} $$X^{\delta, \sigma, \tau}_\tau = 0.$$
\end{remark}

    The processes $(X^{\delta,\sigma,\tau}, Y^{\delta,\sigma,\tau})$ defined in \eqref{eq:explicit-X} turn out to be very useful for our subsequent analysis, as they allow us to identify candidate equilibrium strategies. Specifically, they allow us to introduce the following auxiliary strategies:
    \begin{equation}
        \label{aux-strat}
        \xi^{\delta, \sigma,\tau} = \frac{Y^{\delta, \sigma,\tau} {-\delta\kappa X^{\delta, \sigma,\tau}}   }{\eta}.
    \end{equation}
    For $\delta = 0$ and an exit time $\tau$ the strategy $\xi^{0, \sigma,\tau}$ is the unique optimal trading strategy of the representative agent in a liquidation model without trading constraints and with trading interval $[\sigma,\tau]$. 
    
    For $\delta = \frac 1 N$ and an exit time $\tau$ the strategy is admissible in an $N$-player game without trading constraints and with trading interval $[\sigma,\tau]$ as stated in the following lemma. The proof of (i) follows from \cite[Lemma 2.8]{FHH-2023}; part (ii) follows by construction.

    \begin{lemma}
        \label{lem-admis}
        \begin{itemize}
            \item[ (i)]  {The strategy $ {\xi^{\delta, \sigma, \tau}}$ defined in \eqref{aux-strat} is absolutely continuous on $[\sigma,\tau]$} and there exists a constant $C>0$ that depends only on $\mu, \sigma, \tau, \eta, \lambda, \kappa$ such that
            \begin{align}
                \label{eq:xi_apriori_estiate}
                \Vert \xi^{\delta, \sigma, \tau}\Vert_{\infty} + \Vert \dot\xi^{\delta, \sigma, \tau}\Vert_\infty \le C(1+\abs{x}), \qquad x\in \bR, \; \delta\in[0,1].
            \end{align}
            \item[(ii)]If $\mu\in L^{1}([0,T])$, then the strategy is square integrable on $[\sigma,\tau]$. If, in addition, $\tau$ is an exit time for this strategy, then $ {\xi^{\delta, \sigma, \tau}}$ is admissible.
        \end{itemize}
    \end{lemma}

In what follows we are using the auxiliary strategies $\xi^{\delta, \sigma,\tau}$ to heuristically identify a pair $\{\sigma^*,\tau^*\}$ of market entry and exit times such that the system \eqref{eq:XY-i} on $[0,T]$ and the system \eqref{eq:FBODE} on $[\sigma^*,\tau^*]$ coincide.

    \subsubsection{Candidate entry times}
 
 If a player optimally exits the market at time $\tau^*$ and optimally enters the market at some strictly positive time $\sigma^*$, then we expect the player to start trading gradually. That is, we expect that the smooth pasting condition 
    \begin{equation}
        \label{entry}
        Y^{\delta, \sigma^*,\tau^*}_{\sigma^*} -\delta \kappa X^{\delta,\sigma^*,\tau^*}_{\sigma^*}= 0
    \end{equation}
to hold. To identify candidate optimal entry times we introduce the function    
\begin{equation} \label{eq:psi} 
	\psi^{\delta,\tau}_\mu(t) ~:=~ \frac{B^{\delta,\tau}_t  }{A^\delta_t -\delta\kappa }, \quad t \in [0, \tau]
    \end{equation}
    in terms of which we can represent the adjoint processes $Y^{\delta,\sigma,\tau}$ as
    \begin{equation}
        \label{Yst}
        \begin{split}
            Y^{\delta,\sigma,\tau}_t = \big( A^\delta_t - \delta\kappa \big) \big( X^{\delta,\sigma,\tau}_t + \psi^{\delta,\tau}_\mu(t) \big)   +\delta\kappa X_t^{\delta,\sigma,\tau}              , \quad t \in [\sigma,\tau]     .
        \end{split}
    \end{equation}

    We emphasize that the function $\psi^{\delta,\tau}_\mu$ does not depend on the candidate entry time. If this function is invertible, then it follows from the above that 
    \begin{equation}
        \label{invert}
        -x = \psi^{\delta,\tau^*}_\mu(\sigma^*).
    \end{equation}

\begin{remark}
    Since the function $\psi^{\delta,\tau}_\mu$ is positive for all $\tau$, the preceding equation has no solution for sellers. This suggests that sellers never enter seller-dominated market late to avoid future adverse price movements. 
\end{remark}
    
    To obtain a well-defined candidate entry time, we proceed under the following assumption.

    \begin{ass}
        \label{ass:ceof_relation}
        The function $\psi^{\delta,\tau}_\mu$ is strictly decreasing, hence invertible on the interval $[0,\tau]$ for all $\tau \in [0,T]$, $\delta \in [0,1]$.
    \end{ass}

    The following proposition states sufficient conditions on the model parameters that guarantee the strict monotonicity of $\psi^{\delta,\tau}_\mu$ on $[0,\tau]$. The proof is postponed to the Appendix~\ref{app:A}.

    \begin{proposition}
        \label{prop:sufficient_conditions}
        The function $\psi^{\delta,\tau}_\mu$ admits the integral representation
        \[
            \psi^{\delta,\tau}_\mu(t) = \frac{1}{\alpha^\delta_t}  \int_t^{\tau} e^{ -\int_0^s\frac{A^\delta_r}{\eta_r}\,dr    }\kappa\mu_s\,ds,
            \quad t \in [0, \tau]
        \]
        where
        \[
            \alpha^\delta_t := (A^\delta_t - \delta\kappa)e^{-\int_0^t  \frac{A^\delta_r}{\eta_r}\,dr  }.
        \]
        The function $\alpha^\delta$ is strictly positive, bounded, and differentiable on $[0,T]$.
        In particular, for any $\mu$ that satisfies Assumption~\ref{ass:mu_sign} the function
        $\psi^{\delta,\tau}_\mu$ is bounded, differentiable, and strictly positive on $[0,\tau)$.

        Moreover, the function $\psi^{\delta,\tau}_\mu$ is invertible on $[0,\tau]$ under any of the following conditions: %
        \begin{enumerate}
            [label=(\roman*)]
            \item The function $\lambda$ is small enough (e.g., $\lambda = 0$),
            \item The product $\lambda\eta$ is non-decreasing (e.g., $\lambda$ and $\eta$ are constants).
        \end{enumerate}
    \end{proposition}

    \subsubsection{Candidate exit times}\label{sec:verification-sellers}

    Candidate entry times were identified %
    by the smooth pasting condition of setting the trading rate at a time of late entry to zero. This approach does not carry over to exit times as the corresponding equation $$Y^{\delta,\sigma, \tau}_\tau {-\delta\kappa X^{\delta,\sigma,\tau}_\tau} = 0$$ holds for {\sl any} exit time. In fact, if $\tau < T$, then
    \[
        Y^{\delta,\sigma, \tau}_\tau {-\delta\kappa X^{\delta,\sigma,\tau}_\tau} = (A^\delta_\tau -\delta\kappa) \left( X^{\delta,\sigma, \tau}_\tau + \psi_\mu^{\delta,\tau}(\tau) \right)   = (A^\delta_\tau-\delta\kappa) \left( 0 - 0 \right)   = 0.
    \]
	
	Instead, we proceed as follows. Regardless of the entry time, we expect the auxiliary portfolio process $X^{\delta,\sigma,\tau}$ introduced in \eqref{eq:explicit-X} 
to satisfy
    \begin{equation}
        \label{cand-time2}
        X_{\tau^*}^{\delta,\sigma,\tau^*} = 0
    \end{equation}
    at any exit time $\tau^*$. For $\tau^*< T$, that is in the case of early exit, this implies that 
    \begin{equation}
        \label{eq:=x}
        \int_\sigma^{\tau^*}\frac{1}{\eta_s} e^{\int_\sigma^s\frac{A^{\delta}_r-\delta\kappa}{\eta_r}\,dr} \int_s^{\tau^*}\kappa\mu_u e^{-\int_s^u\frac{A^{\delta}_r}{\eta_r}\,dr}\,du\,ds=x.  %
    \end{equation}

    To identify those initial positions for which early liquidation may take place we introduce the function
    \begin{equation}\label{eq:def_h}
       h^{\delta,\sigma}_t:=e^{    -\int_\sigma^t\frac{A^{\delta}_r}{\eta_r}\,dr    }\int_\sigma^t \frac{1}{\eta_s} e^{ \int_\sigma^s\frac{2A^{\delta}_r - \delta\kappa}{\eta_r}\,dr}\,ds
    \end{equation}
    and apply Fubini's theorem to rewrite the left hand side of the equation \eqref{eq:=x} as
    \begin{equation}
        \label{cha-tau}
        \begin{split}
            & \int_\sigma^t\frac{1}{\eta_s}e^{\int_\sigma^s\frac{A^{\delta}_r-\delta\kappa}{\eta_r}\,dr} \int_s^t \kappa\mu_u e^{    -\int_s^u\frac{A^{\delta}_r}{\eta_r}\,dr    }\,du\,ds \\
            =&~\int_\sigma^t\kappa\mu_u e^{    -\int_\sigma^u\frac{A^{\delta}_r}{\eta_r}\,dr    }\int_\sigma^u \frac{1}{\eta_s} e^{\int_\sigma^s\frac{2A^{\delta}_r - \delta\kappa}{\eta_r}\,dr}\,ds\,du \\
            =& ~  \int_\sigma^t\kappa\mu_u  {h^{\delta,\sigma}_u}\,du  \\
            =: & ~ \phi^{\delta, \sigma}_\mu(t).
        \end{split}
    \end{equation}

    The function $\phi^{\delta, \sigma}_\mu$ is well defined and {\sl strictly} increasing, due to \cite[Lemma 2.6]{FHH-2023}. Hence we expect the first (and optimal) exit time to be characterized by the equation
    \begin{equation}\label{eqn:phi}
    	x=\phi^{\delta,\sigma}_\mu(\tau^*).
    \end{equation}
 
 \begin{remark}
    Since the function $\phi^{\delta,\sigma}_\mu$ is increasing and starts at zero, this equation has no solution for buyers (negative initial positions). This suggests that buyers never exit a seller dominated market early. 
\end{remark}

\subsection{Verification}\label{Sec:Verification}

Our preceding heuristics suggests that a late optimal entry time $\sigma^*$ and an early optimal exit time $\tau^*$ satisfy the equations
\[
    	\psi^{\delta, \tau^*}_\mu(\sigma^*) = -x, \quad \mbox{respectively}, \quad \phi^{\delta, \sigma^*}_\mu(\tau^*) = x.
    \]
    Since the first equation has no solution for sellers, and the second equation has no solution for buyers, we expect that in seller-dominated markets $\sigma^* = 0$ for any seller while $\tau^* = T$ for all buyers. 
    
    In this section we show that this is indeed the case. For sellers we can draw on earlier work; the verification for buyers is more subtle.

    \subsubsection{Sellers}\label{sec:verification-sellers}

Our previous analysis suggests that for any exogenous trading rate $\mu$ that does not change its sign the optimal entry time for a seller with initial endowment $x$ is given by $\sigma^*=0$, that an optimal exit time for sellers is given by 
    \begin{equation}
        \label{eqn:tau}
        \tau_\mu(x) := \inf\left\{ t \in [0,T] : \phi^{\delta,0}_\mu(t) = x\right\} \quad \mbox{with} \quad \inf \emptyset := T
    \end{equation}
and that an optimal trading strategy is given by
        \begin{equation}
            \label{xi*3}
            \xi^{*,\delta,0,x,\mu}_t:=\left\{\begin{matrix}
                                               \frac{Y^{\delta,0,\tau_\mu(x)}_t - \delta \kappa X^{\delta,0,\tau_\mu(x)}_t}{\eta_t} & \textrm{ if }t\in[0, \tau_\mu(x)] \\ 0 & \textrm{ else }
            \end{matrix}\right. .
        \end{equation}
It has been shown in \cite{FHH-2023} that this is indeed the case under a market drop-out assumption where a player drops out of the market the first time her portfolio process hits zero. Since the drop-out constraint is weaker than the ``no change of trading condition'', this shows that $\tau_\mu(x)$ is admissible and hence optimal in a model with trading constraints, provided that the process
    \[
        Y^{\delta,0,\tau_\mu(x)}_t = A^\delta_t X^{\delta,0,\tau_\mu(x)}_t + B^{\delta,\tau_\mu(x)}_t, \quad t \in [0,\tau_\mu(x))
    \]
    is strictly positive, which follows from the strict positivity of the processes $A^\delta$ and $B^{\delta,\tau_\mu(x)}$. Hence, we have shown the following result.

    \begin{theorem}\label{prop:sellers}
        In a seller dominated market with an exogenous impact function $\mu$ that does not change its sign, the optimal exit time is given by \eqref{eqn:tau} and the unique optimal trading strategy is given by \eqref{xi*3}. 
    \end{theorem}

    \subsubsection{Buyers}

For buyers the analysis is more subtle. Our heuristic suggests to consider the exit time $\tau^*=T$, the entry time 
    \begin{equation}
        \label{eqn:sigma}
        \sigma_\mu(x) := \inf\left\{ t \in [0,T] : \psi^{\delta,T}_\mu(t) = -x\right\} \quad \mbox{with} \quad \inf \emptyset := 0
    \end{equation}
and the trading strategy  
 \begin{equation}
            \label{xi-sigma}
            \xi^{*,\delta,T,x,\mu}_t:=\left\{\begin{matrix}
                 \frac{Y^{\delta,\sigma_\mu(x),T}_t  - \delta\kappa X^{\delta,\sigma_\mu(x),T}_t  }{\eta_t} & \textrm{ if }t\in[\sigma_\mu(x),T] \\ 0 & \textrm{ else }
            \end{matrix}\right. .
        \end{equation}

    The following lemma shows that the strategy $\xi^{*,\delta,T,x,\mu}$ is admissible; the proof is given in Appendix A. 

    \begin{lemma}\label{lem:admiss}
        If the function $\psi^{\delta,T}_\mu$ is strictly decreasing, then the process $$Y^{\delta,\sigma_\mu(x), T} - \delta\kappa X^{\delta,\sigma_\mu(x),T} $$ is strictly negative on the interval $(\sigma_\mu(x),T]$. In particular, the strategy \eqref{xi-sigma} is admissible. 
    \end{lemma}

    We emphasize that we do not expect $\xi^{*,\delta,T,x,\mu}$ to be a best response to any given $\mu$. This will only be the case in the MFG and in equilibrium whenever the set of players is finite. 
    To prove our verification result for buyers we fix an initial position $x_i < 0$ of player $i$ and put
    \[
        \xi^{*,i}:= \xi^{*,\frac 1 N, T,x_i, \mu}, \quad X^{*,i} = X^{\frac 1 N, \sigma_\mu(x_i), T}, \quad Y^{i} = Y^{\frac 1 N, \sigma_\mu(x_i), T}, \quad \sigma^{*,i} = \sigma_{\mu}(x_i).
    \]
    
    To substantiate our intuition that $\xi^{*,i}$ being a best response in equilibrium, we further fix a strategy profile $\xi^{-i}=(\xi^j)_{j\neq i}$ of the player's opponents such that the aggregate trading rate equals $\mu$: 
    \begin{align}\label{eq:fixed_profile}
            \frac{1}{N} \sum_{j \neq i} \xi^j + \frac 1 N \xi^{*, i} = \mu.
    \end{align}

    The MFG corresponds to the case $N=\infty$. In this case, the above equality is to be understood as fixing the exogenous trading rate equal to $\mu$ and we set 
    \[
        J_i(\xi;\xi^{-i}) = J_i(\xi;\mu).
    \]
We are now ready to state our verification result for buyers. The proof is given in Appendix A. 

    \begin{theorem}
        \label{thm:veri-N}
        Let $\xi^{-i}$ be a  strategy profile satisfying \eqref{eq:fixed_profile}. Then under Assumption~\ref{ass:mu_sign} and \ref{ass:ceof_relation} the strategy $\xi^{*,i}$ is the unique solution to the optimal control problem
        \begin{align}
            \label{eq:problem-player-very}
            \inf_{\xi \in \mathcal{A}_{x_i}}J(\xi;\xi^{-i}),  \qquad
            X_t = x_i - \int_0^{t}\xi_s \, ds, \qquad t\in[0,T].
        \end{align}
    \end{theorem}

    \section{Equilibrium analysis}\label{sec:eqana}

    In this section, we establish existence and uniqueness of equilibrium results for both the $N$-player game and the corresponding MFG within a common mathematical framework. To this end, we denote by $\xi^{\delta, \mu}=\big(\xi^{*,\delta,x,\mu} \big)_{x\in\mathbb R}$ the vector of optimal trading strategies for buyers and sellers given in \eqref{xi-sigma} and \eqref{xi*3}, respectively, and introduce the mapping
    \[
        F: L^1([0,T]) \to \mathbb{R}^{[0,T]}, \qquad  F(\mu)_t := \int_{\bR}\xi^{*,\delta, x, \mu}_t \nu(dx)
    \]
    that maps exogenous trading rates into an aggregate best response. We expect any fixed-point of the mapping $F$ that does not change its sign to yield a Nash equilibrium.
    This suggests that our trading games can be solved as follows:
    \begin{align*}
        \label{equlibrium_framwork}
        \renewcommand{\arraystretch}{1.5}
        \begin{array}{l}
            1. \text{ Fix } \mu\in L^{1}([0,T]).                                                                            \\
            2. \text{ Consider the candidate strategy profile $\xi^{\delta,\mu}$ for $\delta=0$, resp. $\delta=\frac 1 N$}. \\
            3. \text{ Find the fixed-points $\mu^*$ of the mapping  $\mu \mapsto F(\mu)$ in $L^{1}([0,T])$}.\\
            4. \text{ Verify that $\xi^{\delta,\mu^*}$ is a Nash equilibrium.} %
        \end{array}
    \end{align*}

    We characterize fixed-points of the mapping $\mu \mapsto F(\mu)$ as solutions to a non-standard, path-dependent integral equation with endogenous terminal condition and prove that any solution to that integral equation satisfies Assumption \ref{ass:mu_sign}. In hindsight, this justified the analysis of Section 3. 

{\color{black}    
    The fixed-point equation is derived in Section 4.1. The explicit representation is given in Section 4.1.1; Section 4.1.2 verifies that any solution to our equation does not change its sign. The key challenge is to identify the terminal equilibrium trading rate. This is achieved in Section 4.1.3.

	The fixed-point equation is solved in Section 4.2. The key observation is that the integral equation is equivalent to a three-dimensional ODE with two free parameters and that solving the equation is equivalent to solving a two-dimensional root-finding problem. The first parameter corresponds to the aggregate trading rate at the terminal time as in \cite{FHH-2023}. The second parameter is new; heuristically, it describes the maximal initial endowment that a seller that exits the market early can hold. 

    Finally, in Section~\ref{sec:proof_main}, we bring together the results of the fixed-point analysis and the verification results from Section~\ref{Sec:Verification} to present the proof of the main result, Theorem~\ref{thm:main}.}

    \subsection{The integral equation}\label{sec:equ-eq-MFG}

   {\color{black} To guarantee that the fixed-point mapping is well defined, we recall that we are working under the following assumption on the initial distribution $\nu$ of the players' portfolios.
    \begin{ass}
        \label{ass:initial_distribution} The distribution of initial position $\nu$ has a finite absolute first moment, that is, $$\int |x| \nu(dx) < \infty.$$
    \end{ass}}

    \subsubsection{Representation of fixed-points}

    To derive a more explicit form of the fixed-point mapping, we recall the definitions of the functions
    {
    \begin{equation*}
        \begin{split}
            \phi_\mu(t) &= \int_0^t\kappa\mu_uh^{\delta}_u\,du, \qquad {h^{\delta}_t := h^{\delta,0}_t }= e^{    -\int_0^t\frac{A^{\delta}_r}{\eta_r}\,dr    }\int_0^t \frac{1}{\eta_s} e^{ \int_0^s\frac{2A^{\delta}_r - \delta\kappa}{\eta_r}\,dr}\,ds,\\
            \psi_\mu(t) &= \frac{1}{\alpha^\delta_t}  \int_t^{T} e^{ -\int_0^s\frac{A^\delta_r}{\eta_r}\,dr    }\kappa\mu_s\,ds,
        \end{split}\qquad \begin{split}
           \bigg. \\\bigg. t \in [0, T],
        \end{split}
    \end{equation*}}
    introduced in \eqref{eq:psi} and \eqref{cha-tau} and denote by
    \[
        I_\mu(t) := (-\infty, -\psi_\mu(t)]\cup[\phi_\mu(t), \infty)
    \]
    the set of player types that are active on the market at time $t \in [0,T]$. The following representation of the mapping $F$ will allow us to characterize equilibrium trading rates in terms of integral equations.

    \begin{lemma}
        \label{lem:mu_differentiability}
        For any $\mu\in L^{1}([0,T])$ it holds for all $t \in [0,T]$ that
        \begin{align}
            \label{eq:consitency2}
            F(\mu)_t ~= F(\mu)_T + \int_t^{T}\frac{1}{\eta_s}\int_{I_\mu(s)}\Big(\kappa\mu_s+\lambda_s X^{\delta, x, \mu}_s + (\dot\eta_s-\delta\kappa)\xi^{\ast,\delta,x,\mu}_s\Big) \nu(dx)\,ds,
        \end{align}
        {where $X^{\delta,x,\mu}$ is the state process corresponding to $\xi^{*,\delta,x,\mu}$. }
        In particular, $F$ maps the set $L^{1}([0,T])$ into the space of absolutely continuous functions on $[0,T]$.
    \end{lemma}
    \begin{proof}
        From the definition of the strategies $\xi^{\ast, \delta, x, \mu}$  and the interval $I_\mu(t)$ it follows that
        \[
            \dot\xi^{\ast, \delta, x, \mu}_t = \xi^{*, \delta,x,\mu}_t = 0, \quad x\in (-\psi_\mu(t), \phi_\mu(t)).
        \]
        In view of Lemma \ref{lem-admis} (i) and \cite[Lemma 2.8]{FHH-2023} the optimal strategies are almost everywhere differentiable and the derivative is at most of linear growth in the initial position, uniformly in time. The moment condition on the initial distribution thus allows us to apply Fubini's theorem to the integral representation of $F(\mu)_t$ to deduce that
        \begin{align*}
            \int_{\bR}\xi^{\ast, \delta, x, \mu}_t\nu(dx) &=~ \int_{\bR}\left( \xi^{\ast, \delta, x, \mu}_T - \int_t^{T}\dot \xi^{\ast, \delta, x, \mu}_s\,ds \right) \,\nu(dx)   \\
            &=~ \int_{\bR}\xi^{\ast, \delta, x, \mu}_T\nu(dx)  - \int_t^{T}\int_{I_\mu(s)}\dot \xi^{\ast, \delta, x, \mu}_s\,\nu(dx)  \,ds \\
            &=~F(\mu)_T - \int_t^{T}\int_{I_\mu(s)}\frac{d}{ds}\left(\frac{Y^{\delta, x, \mu}_s-\delta\kappa X^{\delta,x,\mu}_s }{\eta_s}\right)\,\nu(dx)  \,ds, \qquad t\in [0,T].
        \end{align*}
        The assertion now follows by using that $(X^{\delta, x, \mu}, Y^{\delta, x, \mu})$ solves the forward-backward equation \eqref{eq:FBODE}.
    \end{proof}

    Using similar arguments as in the proof of the above lemma it follows that aggregate stock holdings can be represented as
    \begin{align*}
        \int_{I_\mu(t)}X^{\delta, x, \mu}_t \,\nu(dx) &=
        \int_{\bR}X^{\delta, x, \mu}_t \,\nu(dx) - \int_{-\psi_\mu(t)}^{\phi_\mu(t)}X^{\delta, x, \mu}_t \,\nu(dx)\\
        &=  \int_{t}^{T} \int_{\bR}\xi^{\delta, x, \mu}_s\, \nu(dx) \, ds - \int_{-\psi_\mu(t)}^{0}x\, \nu(dx)\\
        &=\int_t^{T}F(\mu)_s \,ds + \ell(-\psi_\mu(t)),
    \end{align*}
    where
    \begin{equation}
        \label{ell}
        \ell(x) :=-\int_x^0 y\, \nu(dy), \qquad x\le0.
    \end{equation}

    In terms of the tail probability functions
    \begin{equation}
        \label{tails}
        \begin{split}
            p:\bR\to[0,1],& \quad x\mapsto\nu((-\infty, x]), \\
            q:\bR\to[0,1],& \quad x \mapsto\nu([x,\infty)),
        \end{split}
    \end{equation}
    the equation \eqref{eq:consitency2} can hence be rewritten as
    \begin{equation}
        \label{eq:consitency3}
        \begin{split}
            F(\mu)_t  =&~ F(\mu)_T + \int_{t}^{T}\frac{\kappa}{\eta_s}\left(q\big( \phi_\mu(s)\big) + p(-\psi_\mu(s) \big) \right)\mu_s\,ds
            + \int_t^{T}\frac{\dot\eta_s-\delta\kappa}{\eta_s}F(\mu)_s\,ds  \\
            &~+  \int_t^{T}\frac{\lambda_s}{\eta_s}\left( \ell(-\psi_\mu(s))+\int_s^{T}F(\mu)_u\,du \right)\,ds, \qquad t\in[0,T].
        \end{split}
    \end{equation}

 {\color{black} 
We note that $q(\phi_\mu(t))$ specifies the number of sellers remaining in the market by time $t$. Likewise, $p(-\psi_\mu(t))$ specifies the number of buyers that have entered the market, while $\ell(-\psi_\mu(t))$ specifies the liquidity that these buyers will bring to the market in the future. In the absence of a trading constraint on buyers, we may set  $\psi_\mu\equiv 0$ and $\ell\equiv 0$, which reduces the equation to \cite[Equation (3.6)]{FHH-2023}.
}

    The proof of the following fixed-point representation is identical to the one in market drop-out model considered in \cite[Proposition 3.3]{FHH-2023}.

    \begin{proposition}
        \label{prop:fix-point-integral} A process $\mu\in L^{1}([0,T])$ solves the fixed-point of $F$ if and only if $\mu_T = F(\mu)_T$ and $\mu$ solves the equation
        \begin{equation}
            \label{eq:mu-first-integral}
            \begin{split}
                \mu_t  =&~ \mu_T + \int_{t}^{T}\frac{\kappa}{\eta_s}\left(q\big( \phi_\mu(s)\big) + p(-\psi_\mu(s) \big) -\delta  \right)\mu_s\,ds
                + \int_t^{T}\frac{\dot\eta_s}{\eta_s}\mu_s\,ds  \\
                &~+  \int_t^{T}\frac{\lambda_s}{\eta_s}\left( \ell(-\psi_\mu(s)) +\int_s^{T}\mu_u \,du \right)\,ds, \qquad t\in[0,T].
            \end{split}
        \end{equation}
    \end{proposition}

    We emphasize that equation \eqref{eq:mu-first-integral} is not a backward equation, due to the dependence of the function $\phi_\mu$ on the forward dynamics of the process $\mu$. 
   
    \subsubsection{The sign condition}

    The following result shows that any fixed-point $\mu$ of the mapping $F$ does not change its sign and that the mapping $t \mapsto \eta_t\mu_t$ is monotone. This justifies our Assumption \ref{ass:mu_sign}, which was key to the analysis of the best response functions carried out in Section \ref{Sec3-new}.

    \begin{lemma}
        \label{lem:apriori}
        Let $\mu$ be a solution to \eqref{eq:mu-first-integral}.
        Then it holds for $\delta \in \{0, \frac 1 N \}$ that
        $$\sign(\mu_t)=\sign(\mu_T),\qquad 0\leq t\leq T.$$
        Furthermore, for $\mu_T > 0$ (resp. $\mu_T <0$) the mapping $t \mapsto \eta_t\mu_t$ is decreasing (resp. increasing).
    \end{lemma}
    \begin{proof}
        Noting that $\psi_\mu(t)\leq \frac{\kappa}{\alpha^\delta_T}\int_t^T|\mu_s|\,ds$, it follows from equation \eqref{eq:mu-first-integral} that there exists a constant $K>0$ depending only on $T, \kappa, \eta$ and $\lambda$ such that
        \begin{equation*}
            |\mu_t| \leq |\mu_T| + K \int_{t}^{T}|\mu_s|\, ds
        \end{equation*}
        for all $t\in[0,T]$, and hence from Gr\"onwall's inequality that
        \[
            |\mu_t|\leq |\mu_T| e^{ K(T-t) }.
        \]

        In particular, $\mu_T=0$ implies that $\mu \equiv 0.$ If $\mu_T > 0$, then it follows from differentiating \eqref{eq:mu-first-integral} that
        \begin{equation*}
            \begin{split}
                \dot\mu_t =&~ -\frac{\kappa}{\eta_t}\left(q\big( \phi_\mu(t)\big) + p(-\psi_\mu(t) \big) -\delta \right)\mu_t
                - \frac{\dot\eta_t}{\eta_t}\mu_t - \frac{\lambda_t}{\eta_t}\left(\ell(-\psi_\mu(t)) +\int_t^{T}\mu_u \,du \right),
            \end{split}
        \end{equation*}
        for almost every $t\in[0,T]$.
        We denote the first time after which $\mu$ stays positive by
        \[
            t_0 := \inf\left\{t\in[0, T] \;\big\vert\; \mu\vert_{[t,T]} > 0\right\}.
        \]
        In particular, from that time forward a strictly positive proportion of sellers is trading the stock. Hence,
        \[
            q\big( \phi_\mu \big) + p(-\psi_\mu \big) - \delta \geq 0 \quad \mbox{on} \quad (t_0,T].
        \]
        As a result,
        \begin{equation*}
            \begin{split}
                \dot\mu <&~ - \frac{\dot\eta}{\eta}\mu \quad \mbox{on} \quad [t_0,T].
            \end{split}
        \end{equation*}
        In particular, the function $\mu \cdot \eta$ is strictly decreasing on $[t_0, T]$ and so
        \[
            \mu_t > \frac{\eta_T}{\eta_t}\mu_T > 0 \quad \mbox{on} \quad t \in [t_0, T].
        \]
        By continuity of $\mu$ it must thus hold that $t_0 = 0$.
        The case $\mu_T <0$ follows analogously.
    \end{proof}

    \subsubsection{The terminal condition}\label{sec:terminal-MFG}

    Having derived a characterization of the fixed-points in terms of a non-linear integral equation, the following proposition identifies the terminal condition of the integral equation and determines its sign. It turns out that the terminal condition depends on the proportion of sellers that do not exit the market early as well as on the entire history of market entries.

    In what follows we denote by $\E[\nu]$ the first moment of the distribution $\nu$.

    \begin{proposition}
        \label{prop:terminal-MFG}
        Let Assumption \ref{ass:ceof_relation} hold.
        A function $\mu \in L^1([0,T])$ is a fixed-point of the mapping $F$ if and only if it satisfies the integral equation \eqref{eq:mu-first-integral} and the implicit terminal condition
        \[
            \mu_T= \frac{\widetilde \alpha^\delta_T}{\eta_T}  \left\{     \mathbb E[\nu]  -Q(\phi_\mu(T)) -P(-\psi_\mu(0))  + \int_{0}^T p(-\psi_\mu(t)) \frac{1}{\widetilde \alpha^\delta_t} (\lambda_t\psi_\mu(t)-\kappa \mu_t)\,dt \right\},
        \]
   where
        \[
            \widetilde\alpha^\delta_t :=(A^\delta_t - \delta\kappa) e^{ -\int_0^t\frac{A^\delta_r-\delta\kappa}{\eta_r}\,dr },
            \quad Q(x) :=\int_0^x q(y)\,dy,
            \quad P(x) :=-\int_{x}^{0} p(y) dy.
        \]

        The terminal condition can be equivalently written as
        \begin{equation}
            \label{eq:terminal_condtion}
            \begin{split}
                \mu_T
                =&~ \frac{\widetilde\alpha^\delta_T}{\eta_T} \left\{\mathbb E[\nu]  - Q(\phi_\mu(T) ) + \int_{0}^T P(-\psi_\mu(t))e^{  \int_0^{t} \frac{A^\delta_r - \delta \kappa}{\eta_r}\,dr} \frac{A^\delta_t}{\eta_t}\,dt  \right\}.
            \end{split}
        \end{equation}
        In particular, for any fixed-point $\mu \in L^1([0,T])$ it holds that
        $$\sign(\mu_t) = \sign(\E[\nu]), \qquad 0 \le t \le T.$$
    \end{proposition}

    \begin{proof}
        We proceed in two steps, starting with the characterization of the terminal value. We assume w.l.o.g.~that $\mu_T > 0$ and set
        \[
            a:=-\psi_\mu(0), \quad b:=\phi_\mu(T), \quad \sigma:=\sigma_\mu(x).
        \]

            {\bf Step 1. Characterization of $\mu_T$}.
        Taking limits in the fixed-point equation we obtain that
        \begin{equation*}
            \begin{split}
                \mu_T= &~   \lim_{t\nearrow T}\mu_t \\ =&~ \lim_{t\nearrow T} \int_{\mathbb{R}}\xi^{\ast, \delta,  x,\mu}_t\,\nu(dx) \\
                =&~ \lim_{t\nearrow T} \int_{-\infty}^{-\psi_\mu(t)} \xi^{\ast, \delta, x,\mu}_t\,\nu(dx) +  \lim_{t\nearrow T} \int_{\phi_\mu(t)}^\infty \xi^{\ast, \delta, x,\mu}_t \nu(dx)\\
                =&~ \lim_{t\nearrow T} \int_{-\infty}^{-\psi_\mu(t)} \frac{(A^\delta_t-\delta\kappa) X^{\delta,x,\mu}_t}{\eta_t}\,\nu(dx)  +  \lim_{t\nearrow T} \int_{\phi_\mu(t)}^\infty    \frac{(A^\delta_t - \delta\kappa) X^{\delta,x,\mu}_t}{\eta_t} \nu(dx)\\
                :=&~ I_1+I_2.
            \end{split}
        \end{equation*}

        The same calculation as in the proof of \cite[Proposition 3.5]{FHH-2023} shows that the second term is given by
        \[
            I_2= \frac{\widetilde \alpha^\delta_T}{\eta_T} \left( \int_{b}^\infty x\,\nu(dx)   -b \int_{b}^\infty \nu(dx)    \right).
        \]

        The first term captures the impact of buyers on the terminal trading rate. It satisfies
        \begin{align*}
            I_1=&~\int_{-\infty}^0 \frac{\lim_{t\nearrow T }(A^\delta_t -\delta\kappa)X^{\delta,x,\mu}_t }{\eta_T} \,\nu(dx)\\
            =&~ \frac{1}{\eta_T}\int_{-\infty}^0   \lim_{t\nearrow T }    (A^\delta_t -\delta\kappa) \left(            xe^{-\int_{\sigma }^t \frac{A^\delta_r-\delta\kappa}{\eta_r}\,dr} 
            -\int_{\sigma }^t \frac{1}{\eta_s} e^{-\int_s^t \frac{A^\delta_r-\delta\kappa}{\eta_r}\,dr }\int_s^{T} \kappa \mu_u e^{-\int_s^u\frac{A^\delta_r}{\eta_r}\,dr}\,du\,ds                    \right)         \nu(dx) \\
            =&~\frac{1}{\eta_T}\int_{-\infty}^0 x \lim_{t\nearrow T }    (A^\delta_t -\delta\kappa)e^{-\int_\sigma^t \frac{A^\delta_r-\delta\kappa}{\eta_r}\,dr} \nu(dx)  \\
            &~  -  \frac{1}{\eta_T}\int_{-\infty}^0 \lim_{t\nearrow T }(A^\delta_t-\delta\kappa)\int_\sigma^t \frac{1}{\eta_s} e^{-\int_s^t \frac{A^\delta_r-\delta\kappa}{\eta_r}\,dr }\int_s^{T} \kappa \mu_u e^{-\int_s^u\frac{A^\delta_r}{\eta_r}\,dr}\,du\,ds                                   \,\nu(dx) \\
            =&~\frac{\widetilde\alpha^\delta_T}{\eta_T}\int_{-\infty}^0 x e^{  \int_0^{\sigma}\frac{A^\delta_r-\delta\kappa}{\eta_r}\,dr }  \,\nu(dx) \\
            &~-  \frac{1}{\eta_T}\int_{-\infty}^0 \lim_{t\nearrow T }(A^\delta_t -\delta\kappa)e^{-\int_0^t \frac{A^\delta_r-\delta\kappa}{\eta_r}\,dr}  \int_\sigma^t \frac{1}{\eta_s} e^{\int_0^s \frac{A^\delta_r - \delta\kappa }{\eta_r}\,dr }\int_s^{T} \kappa \mu_u e^{-\int_s^u\frac{A^\delta_r}{\eta_r}\,dr}\,du\,ds                                   \,\nu(dx) \\
            =&~\frac{\widetilde\alpha^\delta_T}{\eta_T}\int_{-\infty}^0 x e^{  \int_0^{\sigma}\frac{A^\delta_r-\delta\kappa}{\eta_r}\,dr } \,\nu(dx) -  \frac{\widetilde\alpha^\delta_T }{\eta_T}\int_{-\infty}^0    \int_\sigma^T \frac{1}{\eta_s} e^{\int_0^s \frac{A^\delta_r-\delta\kappa}{\eta_r}\,dr }\int_s^{T} \kappa \mu_u e^{-\int_s^u\frac{A^\delta_r}{\eta_r}\,dr}\,du\,ds                                   \,\nu(dx).
        \end{align*}
        Hence, defining
        $$
        g_\mu(t):= \int_t^T \frac{1}{\eta_s} e^{ \int_0^s \frac{A^\delta_r-\delta\kappa}{\eta_r}\,dr }\int_s^{T} \kappa \mu_u e^{-\int_s^u\frac{A^\delta_r}{\eta_r}\,dr}\,du\,ds,
        $$
        we have
        \begin{align*}
            I_1
            &=~\frac{\widetilde\alpha^\delta_T }{\eta_T}  \left(\int_{-\infty}^0 x e^{  \int_0^{\sigma}\frac{A^\delta_r-\delta\kappa}{\eta_r}\,dr } \,\nu(dx) -  \int_{-\infty}^0 g_\mu(\sigma) \,\nu(dx)  \right).
        \end{align*}

        Since $\sigma=0$ for all $x \in (-\infty, -\psi_\mu(0)] = (-\infty, a]$ we see that
        \begin{equation*}
            \begin{split}
                \frac{\eta_T}{\widetilde\alpha^\delta_T}I_{1}  =&~ \int_{a}^0 x e^{  \int_0^{\sigma}\frac{A^\delta_r-\delta\kappa}{\eta_r}\,dr } \,\nu(dx)  +  \int_{-\infty }^{ a } x e^{  \int_0^{0}\frac{A^\delta_r-\delta\kappa}{\eta_r}\,dr } \,\nu(dx) -  \int_{a}^0 g_\mu(\sigma) \,\nu(dx) -  \int_{-\infty}^{a}   g_\mu(0) \,\nu(dx)\\
                =&~\int_{a}^0 \left(x e^{  \int_0^{\sigma}\frac{A^\delta_r-\delta\kappa}{\eta_r}\,dr } - g_\mu(\sigma)\right) \,\nu(dx)  +  \int_{-\infty}^{a} (x-g_\mu(0))   \,\nu(dx).
            \end{split}
        \end{equation*}

        In terms of the tail probabilities $p$ and $q$ introduced in \eqref{tails} and using that  $g_\mu(0)=\phi_\mu(T)=b$ the terminal condition can hence be represented as follows:
        \begin{align*}
            \mu_T=&~\frac{ \widetilde \alpha^\delta_T  }{\eta_T} \left\{   \int_{a}^0 \left(x e^{  \int_0^{\sigma}\frac{A^\delta_r-\delta\kappa}{\eta_r}\,dr } - g_\mu(\sigma)\right)\nu(dx)       +  \int_{\bR\setminus[a,b]}x   \,\nu(dx) - b(p(a) + q(b))   \right\}.
        \end{align*}

        We use an integration by parts argument to simplify the first term.
        Since $\dot\psi_\mu < 0$ on $[0,T]$  the entry time $\sigma = \sigma_\mu(x) =  \psi_\mu^{-1}(-x)$ is differentiable on $[-\psi_\mu(0), \psi_\mu(T)]=[a,0]$ and
        $$\frac{d}{dx} g_\mu(\sigma) = -\frac{(A^\delta_\sigma-\delta\kappa) e^{\int_0^{\sigma}\frac{A^\delta_r-\delta\kappa}{\eta_r}dr}}{\eta_\sigma}\psi_\mu(\sigma)\dot\sigma = -\frac{(A^\delta_\sigma - \delta\kappa) e^{\int_0^{\sigma}\frac{A^\delta_r-\delta\kappa}{\eta_r}dr}}{\eta_\sigma}\frac{x}{\dot\psi_\mu(\sigma)} .$$
        Using partial integration it follows that
        \begin{equation*}
            \begin{split}
                &~ \int_{a}^0 \left(x e^{  \int_0^{\sigma}\frac{A^\delta_r-\delta\kappa}{\eta_r}\,dr } - g_\mu(\sigma)\right) \,\nu(dx)  \\
                = & \lim_{\varepsilon \to 0} \left\{ \left(x e^{  \int_0^{\sigma}\frac{A^\delta_r-\delta\kappa}{\eta_r}\,dr } - g_\mu(\sigma)\right)p(x)\right\}_{x = a}^{x=\varepsilon} - \int_{a}^0 p(x) e^{  \int_0^{\sigma} \frac{A^\delta_r-\delta\kappa}{\eta_r}\,dr} \,dx.
            \end{split}
        \end{equation*}

        Applying L'Hôpital's rule we further obtain that
        \begin{align*}
            \lim_{\varepsilon \to 0}  \left(\varepsilon  e^{  \int_0^{\sigma(\varepsilon)}\frac{A^\delta_r-\delta\kappa}{\eta_r}\,dr } - g_\mu(\sigma(\varepsilon))\right)p(\varepsilon)
            &= \lim_{t \to T}  \left(-\psi_\mu(t)  e^{  \int_0^{t}\frac{A^\delta_r-\delta\kappa}{\eta_r}\,dr } - g_\mu(t)\right)p(-\psi(t)) \\
            &=   -p(0)\lim_{t \to T} \psi_\mu(t)  e^{  \int_0^{t}\frac{A^\delta_r-\delta\kappa}{\eta_r}\,dr } \\
            &=   p(0)\lim_{t \to T} \frac{\frac{\eta_t}{A^\delta_t-\delta\kappa}(\lambda_t \psi_\mu(t) - \kappa \mu_t)}{ \widetilde\alpha^\delta_t}  \\
            &= 0.
        \end{align*}

        Inserting the above calculations and summarizing the remaining integral terms yields that
        \begin{align*}
            \mu_T =&~ \frac{\widetilde\alpha^\delta_T}{\eta_T} \left\{   -(a-b) p( a) + \int_{\bR\setminus[a,b]}x   \,\nu(dx) - b(p(a) + q(b)) -\int_{a}^0 p(x) e^{  \int_0^{\sigma} \frac{A^\delta_r-\delta\kappa}{\eta_r}\,dr} \,dx  \right\}.
        \end{align*}

        The substitution $x = -\psi_\mu(t)$ simplifies the second integral term to
        \begin{equation}
            \label{eq:proof_terminal_ibp}
            \begin{split}
                \int_{a}^0 p(x) e^{  \int_0^{\sigma} \frac{A^\delta_r-\delta\kappa}{\eta_r}\,dr} \,dx  =& -\int_{0}^T p(-\psi_\mu(t)) e^{  \int_0^{t} \frac{A^\delta_r-\delta\kappa}{\eta_r}\,dr} \dot\psi_\mu(t)\,dt \\
                =& -\int_{0}^T p(-\psi_\mu(t)) \frac{1}{\widetilde\alpha^\delta_t} (\lambda_t \psi_\mu(t)- \kappa \mu_t)\,dt.
            \end{split}
        \end{equation}

        Using that
        \[
            P(x)=xp(x)+\int_x^0 y\nu(dy), \qquad Q(x) = xq(x)+\int_0^{x}y\nu(dy),
        \]
        and summarizing the remaining terms we finally arrive at
        \begin{equation*}
            \begin{split}
                \mu_T
                =&~ \frac{\widetilde\alpha^\delta_T}{\eta_T} \left\{\mathbb E[\nu]  - Q(b) - P(a) + \int_{0}^T p(-\psi_\mu(t)) \frac{1}{\widetilde\alpha^\delta_t} (\lambda_t \psi_\mu(t)-\kappa \mu_t)\,dt  \right\}.
            \end{split}
        \end{equation*}

        {\bf Step 2. Alternative characterization and identification of the sign.} To determine the sign of $\mu_T$ we establish an alternative representation. Applying integration by parts in \eqref{eq:proof_terminal_ibp}, we see that
        \begin{align*}
            \label{eq:proof_terminal_ibp}
            & \int_{a}^0 p(x) e^{  \int_0^{\sigma} \frac{A^\delta_r-\delta\kappa}{\eta_r}\,dr} \,dx \\
            =& -\int_{0}^T p(-\psi_\mu(t)) e^{  \int_0^{t} \frac{A^\delta_r-\delta\kappa}{\eta_r}\,dr} \dot\psi_\mu(t)\,dt \\
            =& \lim_{t \to T}\left\{P(-\psi_\mu(t))e^{  \int_0^{t} \frac{A^\delta_r-\delta\kappa}{\eta_r}\,dr}\right\} - P(a) -\int_{0}^T P(-\psi_\mu(t))e^{  \int_0^{t} \frac{A^\delta_r-\delta\kappa}{\eta_r}\,dr}\frac{A^\delta_t-\delta\kappa}{\eta_t} \,dt \\
            =&  - P(a) -\int_{0}^T P(-\psi_\mu(t))e^{  \int_0^{t} \frac{A^\delta_r-\delta\kappa}{\eta_r}\,dr}\frac{A^\delta_t-\delta\kappa}{\eta_t} \,dt,
        \end{align*}
        where the last equation follows from an application of L'Hôpital's rule. This shows that
        \begin{equation*}
            \begin{split}
                \mu_T
                =&~ \frac{\alpha_T}{\eta_T} \left\{\mathbb E[\nu]  - Q(b) + \int_{0}^T P(-\psi_\mu(t))e^{  \int_0^{t} \frac{A^\delta_r-\delta\kappa}{\eta_r}\,dr}\frac{A^\delta_t-\delta\kappa}{\eta_t}\,dt  \right\}.
            \end{split}
        \end{equation*}

        Let us now assume to the contrary that $\E[\nu] \le 0$. Then the right-hand side of the above equation is non-positive (recall that $P(x) \le 0$ for all $x\le 0$), which contradicts our assumption $\mu_T > 0$. Hence,
        \[
            \sign(\mu_T) = \E[\nu]
        \]
        and by Lemma~\ref{lem:apriori} it follows that $\sign(\mu_t) = \E[\nu]$ for all $0 \le t \le T$. The case $\mu_T < 0$ follows by symmetry. If $\mu_T = 0$, then it follows from Lemma~\ref{lem:apriori} that $\mu \equiv 0$ and, hence that  $a = b = 0$ and $\psi_\mu \equiv 0$. Hence, in this case $\E[\nu] = 0$.
    \end{proof}

    \subsection{Fixed-point analysis}

    Two key challenges arise when solving equation \eqref{eq:mu-first-integral} with the terminal condition \eqref{eq:terminal_condtion}. First, the terminal condition is 
    implicitly given in terms of the entire solution; second, the equation is neither a forward, nor a backward equation due to the dependence of $\phi_\mu(t)$ on the forward path $(\mu_s)_{0 \leq s \leq t}$.

        To overcome both problems, we consider a family of parameterized backward equations subject to a consistency requirement on the parameters. More precisely, we replace the implicit terminal value $\mu_T$ by a generic parameters $\theta \ge 0$ and the endogenous quantity $\phi_\mu(T)$ by a generic parameter $c \geq 0$. The resulting parameterized backward equation reads:
    \begin{equation}
        \label{eq:mu-backwards}
        \begin{split}
            \mu_t  =&~ \theta + \int_{t}^{T}\frac{\kappa}{\eta_s}\left(q\left( c - \int_s^T h^\delta_u \kappa\mu_u du \right) + p(-\psi_\mu(s) ) -\delta \right)\mu_s\, ds
            + \int_t^{T}\frac{\dot\eta_s}{\eta_s}\mu_s\,ds  \\
            &~+  \int_t^{T}\frac{\lambda_s}{\eta_s}\left( \ell(-\psi_\mu(s)) +\int_s^{T}\mu_u \,du \right)\,ds, \qquad t\in[0,T].
        \end{split}
    \end{equation}

    \begin{remark}
        The key difference between the market drop-out model considered in \cite{FHH-2023} and the model considered in this paper is that the terminal condition in \cite{FHH-2023} depends on the trading rate $\mu$ only through the quantity $\phi_\mu(T)$. In that setting, the equilibrium equation could be solved by solving a one-dimensional root finding problem. In our current setting, the terminal condition depends on the entire history of market entries, which renders the root-finding problem much more complex.
    \end{remark}
        In a first step, Section~\ref{sec:auxiliary_backwards:equations}, we prove that for any pair of parameters $(\theta,c)$ there exists a unique solution to the terminal value problem \eqref{eq:mu-backwards}, which we denote by $\mu^{\theta, c}$. To establish the existence and uniqueness of a fixed-point, we then show in a second step, Section~\ref{sec:eu_fixedpoints}, that there exists a pair $(\theta,c)$ such that the following conditions hold:
    \begin{align}
        \label{eq:bc_identities}
        c = \int_0^T h^\delta_s \kappa \mu^{\theta,c}_s\,ds, \quad
        \theta =  \frac{\widetilde\alpha^\delta_T}{\eta_T}\left( \mathbb E[\nu] - Q(c) + \int_0^T e^{\int_0^t\frac{A^\delta_r-\delta\kappa}{\eta_r}dr}\frac{A^\delta_t-\delta\kappa}{\eta_t}P(-\psi_{\mu^{\theta,c}}(t))\, dt\right).
    \end{align}

    The corresponding solution $\mu^{\theta, c}$ yields a solution to our fixed point equation \eqref{eq:mu-first-integral} with terminal condition \eqref{eq:terminal_condtion}, therefore, a fixed point of the mapping $F$.

 In what follows, $K$ denotes a positive constant that may change from line to line, but only depends on $T$ and the parameters $\eta, \kappa$ and $\lambda$. In particular, the involved estimates will hold uniformly in $\delta\in[0,1]$.

\subsubsection{Auxiliary backwards equations}\label{sec:auxiliary_backwards:equations}
    To eliminate the terms resulting from the derivative of $\eta$, it will be convenient to do the  substitutions $\vartheta := \mu \eta$ and $\pi := \frac{\kappa}{\eta}$ in equation \eqref{eq:mu-backwards}, from which we obtain the following modified equation:
        \begin{equation}
            \label{eq:vartheta}
            \begin{split}
                \vartheta_t  &=~ \theta\eta_T + \int_{t}^{T} \left\{q\left( c - \int_s^T h^\delta_u \pi_u \vartheta_u du \right) + p\left(-\frac{1}{\alpha^\delta_s}\int_s^{T} e^{ -\int_0^{u}\frac{A^\delta_r}{\eta_r}\,dr    }\pi_u\vartheta_u\,du\right) -\delta \right\}{\pi_s} \vartheta_s\, ds \\
                &~+  \int_t^{T}{\lambda_s} \left\{ \ell\left(-\frac{1}{\alpha^\delta_s}\int_s^{T} e^{ -\int_0^{u}\frac{A^\delta_r}{\eta_r}\,dr    }\pi_u\vartheta_u\,du\right) +\int_s^{T}\frac{1}{\eta_u}\vartheta_u \,du \right\}\,ds.
            \end{split}
        \end{equation}

    \begin{theorem}
        \label{thm:existence}
        Assume that $\E[\nu] > 0$. For any $\theta, c \ge0$ there exists a unique solution $\vartheta^{\theta,c}$ to the integral equation \eqref{eq:vartheta}.
        In particular, $\mu^{\theta,c} = \eta^{-1} \vartheta^{\theta, c}$ is the unqiue solution to \eqref{eq:mu-backwards}. 

    \end{theorem}

    \begin{proof}
        We prove the existence and uniqueness of solutions to \eqref{eq:vartheta} %
        by separating the linear and non-linear parts.

        \textbf{Step 1. Refomulating the equation.} First write equation \eqref{eq:vartheta} as
        \begin{equation}
            \label{eq:vartheta2}
            \begin{split}
                \vartheta_t  &=~ \mathcal I_\theta + G(\vartheta)_t + H_c(\vartheta)_t + \mathcal J(\vartheta)_t,
            \end{split}
        \end{equation}
        where ${\cal I}_\theta :=\eta_T \theta$ and the functions  $ G, H_c, {\cal J}: C^{0}([0,T]) \to C^{0}([0,T])$ are defined by
        \begin{equation}
            \label{eq:GHJ_def}
            \begin{split}
                G(\vartheta)_t  =&~ \int_t^{T}{\lambda_s}\left(\int_s^{T}\frac{1}{\eta_u}\vartheta_u du\right)ds,\\
                H_{c}(\vartheta)_t  =&~
                \int_{t}^{T} q\left( c - \int_s^T h^\delta_u \pi_u \vartheta_u du \right) {\pi_s} \vartheta_s\, ds\\
                {\cal J}(\vartheta)_t  =&~ \int_{t}^{T} \left\{ p\left(-\frac{1}{\alpha^\delta_s}\int_s^{T} e^{ -\int_0^{u}\frac{A^\delta_r}{\eta_r}\,dr    }\pi_u\vartheta_u\,du\right) -\delta\right\}{\pi_s} \vartheta_s\, ds \\
                &~+ \int_t^{T}{\lambda_s} \ell\left(-\frac{1}{\alpha^\delta_s}\int_s^{T} e^{ -\int_0^{u}\frac{A^\delta_r}{\eta_r}\,dr    }\pi_u\vartheta_u\,du\right)\,ds.
            \end{split}
        \end{equation}
        We deduce the existence and uniqueness of a solution to the equation \eqref{eq:vartheta} from suitable growth and Lipschitz-type estimates on the auxiliary function. To establish Lipschitz estimates for the non-linear functions $H_c$ and ${\cal J}$ we will use the following representations:
        \begin{equation}
            \label{eq:J_integrated_form}
            \begin{split}
                H_c(\vartheta)_t =& ~ \frac{1}{h^\delta_T}Q\left( c\right) - \frac{1}{h^\delta_t}Q\left( c - \int_t^T h^\delta_u \pi_u \vartheta_u \,du \right)
                +\int_{t}^{T}\frac{\dot{h}^\delta_s}{(h_s^\delta)^2}Q\left( c - \int_s^T h^\delta_u \pi_u \vartheta_u \,du \right)\, ds, \\ %
                {\cal J}(\vartheta)_t =&~ - (A^\delta_t-\delta\kappa) P\left(-\frac{1}{\alpha^\delta_t}\int_t^{T} e^{ -\int_0^{u}\frac{A^\delta_r}{\eta_r}\,dr    }\pi_u\vartheta_u\,du\right) \\
                &~- \int_{t}^{T}\left( \frac{(A_s^\delta)^2}{\eta_s} - \frac{\delta\kappa}{\eta_s}A^\delta_s\right) P\left(-\frac{1}{\alpha^{\delta}_s}\int_s^{T} e^{ -\int_0^{u}\frac{A^\delta_r}{\eta_r}\,dr    }\pi_u\vartheta_u\,du\right) \,ds - \int_t^T\delta\pi_s\vartheta_s\,ds.
            \end{split}
        \end{equation}

        The representation of the function $H_c$ follows from integration by parts, noting that by \cite[Lemma~2.6]{FHH-2023} the function $\dot{h}^\delta$ is bounded, that $h^\delta$ is increasing, and that $h^\delta_t > 0$ for all $t > 0$. 
        
        The representation of ${\cal J}$ requires a more intricate analysis at the integration limits. %
         {\color{black} Using that $\ell(x) = xp(x) - P(x)$ for all $x \le 0$ and furthermore that
        $$\dot{\psi}_{\vartheta/\eta}(t) = \frac{1}{A^\delta_t-\delta\kappa}\left(  \lambda_t \psi_{\vartheta/\eta}(t)  - \pi_t\vartheta_t\right), \quad t\in[0,T],$$
        we have that
        \begin{align}\label{eq:J_proof_intermediate}
        {\cal J}(\vartheta)
            _t  =&~ - \int_{t}^{T}(A^\delta_s-\delta\kappa) p(-\psi_{\vartheta/\eta}(s) \big) \dot\psi_{\vartheta/\eta}(s)\, ds - \int_t^T\delta\pi_s\vartheta_s\,ds  - \int_t^{T} \lambda_s P(-\psi_{\vartheta/\eta}(s))ds,
        \end{align}
        for all $t\in[0,T]$.
        We now use integration by parts to obtain that
        \begin{equation}
        \begin{split}\label{eq:proof_integration_by_parts}
            &~-\int_t^{T} (A^\delta_s-\delta\kappa) p(-\psi_{\vartheta / \eta}(s))\dot\psi_{\vartheta / \eta}(s)\,ds\\
            =&~ \left. (A^\delta_s-\delta\kappa)P(-\psi_{\vartheta / \eta}(s))\right\vert_{s=t}^{T} - \int_t^{T}\dot{A}^\delta_s P(-\psi_{\vartheta/\eta}(s)) \,ds \\
            =&~- (A^\delta_t  - \delta\kappa) P(-\psi_{\vartheta / \eta}(t)) - \int_t^{T}\left(\frac{(A^\delta_s)^2}{\eta_s^{}} -\frac{\delta\kappa}{\eta_s}A^\delta_s - \lambda_s  \right)P(-\psi_{\vartheta / \eta}(s))ds, \quad t\in[0,T],
            \end{split}
        \end{equation}
        where we have used that the limit of the first term in the right-hand side is zero as $t\to T$, which can be obtained by using L'Hospital's rule.
        Plugging this into \eqref{eq:J_proof_intermediate} yields the representation in \eqref{eq:J_integrated_form} for ${\cal J}$.}
        
        \textbf{Step 2. Growth and Lipschitz estimates.} In view of the boundedness of the model parameters, the boundedness of the tail probability function $p$, and the linear growth estimate $|\ell(x)| \leq 2 |x|$ for all $x<0$, we get from the definition in \eqref{eq:GHJ_def} that
        \begin{align}
            \label{eq:GHJ_growth}
            \vert G(\vartheta)_t \vert + \vert H_c(\vartheta)_t \vert + \vert {\cal J}(\vartheta)_t \vert  \leq &~ K \int_{t}^{T} \vert\vartheta_s\vert\, ds, \quad t\in[0,T], \quad \vartheta \in C^{0}([0,T]), \quad c\ge 0.
        \end{align}
        
        Since $Q$ is Lipschitz-continuous with coefficient $q(0) \le 1$, we readily deduce that for any $\varepsilon > 0$ there exists a constant $L_\varepsilon > 0$ such that for all $\vartheta, \widetilde{\vartheta} \in C^{0}([0,T])$ and $c, \widetilde{c} \ge 0$ it holds that
        \begin{align}
            \label{eq:H_estimate_lip}
            \vert H_c(\vartheta)_t - H_{\widetilde c}(\widetilde\vartheta)_t\vert \le&~ L_\varepsilon \left( \vert c - \widetilde{c} \vert + \int_t^{T} \Vert \vartheta - \widetilde{\vartheta}\Vert_{\infty;[s, T]} \, ds \right), \quad t\in[\varepsilon,T],
        \end{align}
        where $\| y \|_{\infty;[s,T]}:=\sup_{s\leq r\leq T}|y_r|$.
        The corresponding estimate for the function ${\cal J}$ is more involved. We first notice that $P$ is Lipschitz continuous with constant $p(0) \le 1$.
        Therefore, for any $\vartheta, \widetilde{\vartheta}\in C^{0}([0,T])$ we can estimate
        \begin{align*}
            &~(A^\delta_t-\delta\kappa)\left\vert P\left(-\frac{1}{\alpha^\delta_t}\int_t^{T} e^{ -\int_0^{u}\frac{A^\delta_r}{\eta_r}\,dr    }\pi_u\vartheta_u\,du\right) - P\left(-\frac{1}{\alpha^\delta_t}\int_t^{T} e^{ -\int_0^{u}\frac{A^\delta_r}{\eta_r}\,dr    }\pi_u\widetilde{\vartheta}_u\,du\right)\right\vert \\
            \le&~ \frac{(A^\delta_t-\delta\kappa)}{\alpha^\delta_t}\int_t^{T} e^{-\int_0^{s}\frac{A^\delta_r}{\eta_r}dr}\pi_s \abs{\vartheta_s - \widetilde{\vartheta}_s}\,ds \\
            \le&~ K\frac{(A^\delta_t-\delta\kappa) e^{-\int_0^{t}\frac{A^\delta_r}{\eta_r}dr}}{\alpha^\delta_t}\int_t^{T}\abs{\vartheta_s - \widetilde{\vartheta}_s}\,ds \\
            =&~ K\int_t^{T}\abs{\vartheta_s - \widetilde{\vartheta}_s}\,ds,
        \end{align*}
        for all $t \in [0,T]$.
        Similarly,  for all $t\in[0,T]$ we can estimate the second part of ${\cal J}$ as follows:
        \begin{align*}
            &~\int_t^{T}\frac{(A_s^\delta)^2-\delta\kappa A^\delta_s}{\eta_s}\left\vert P\left(-\frac{1}{\alpha^\delta_s}\int_s^{T} e^{ -\int_0^{u}\frac{A^\delta_r}{\eta_r}\,dr    }\pi_u\vartheta_u\,du\right) - P\left(-\frac{1}{\alpha^\delta_s}\int_s^{T} e^{ -\int_0^{u}\frac{A^\delta_r}{\eta_r}\,dr    }\pi_u\widetilde{\vartheta}_u\,du\right)\right\vert ds \\
            \leq&~ K\int_t^T (A^\delta_s+1)\left(\int_s^T|\vartheta_u -  \widetilde{\vartheta}_u|\,du\right)ds\\
            \leq&~K\int_t^T (A^\delta_s+1) (T-s)\|\vartheta -\widetilde\vartheta\|_{\infty;[s,T]}\,ds\\
            \le& K\int_t^{T} \Vert \vartheta - \widetilde{\vartheta}\Vert_{\infty;[s, T]} \, ds,
        \end{align*}
        where the second to last inequality uses \cite[Lemma~A.1]{FHH-2023}. 
        
        Summarizing the above estimates and using the linearity of $G$, we see that for any $\varepsilon > 0$ there exists a constant $L_\varepsilon > 0$ such that for all $c \ge 0$ and $\vartheta, \widetilde{\vartheta} \in C^{0}([0,T])$ the following holds:
        \begin{align}
            \label{eq:GHJ_lip}
            \vert G(\vartheta)_t - G(\widetilde\vartheta)_t \vert + \vert H_c(\vartheta)_t - H_c(\widetilde\vartheta)_t\vert + \vert {\cal J}(\vartheta)_t - {\cal J}(\widetilde\vartheta)_t \vert  \leq &~ L_{\varepsilon}\int_{t}^{T} \Vert \vartheta - \widetilde{\vartheta}\Vert_{\infty;[s, T]} \, ds, \quad t\in[\varepsilon,T].
        \end{align}
        
         \textbf{Step 3. Constructing the solution.}
        Iterating the estimate \eqref{eq:GHJ_lip} shows that for any $\theta,c \ge 0$ and any $n\in\bN$,
        \begin{align*}
            &\left\Vert [ {\cal I}_\theta + G+ H_c + {\cal J} ]^{n}(\vartheta) - [ {\cal I}_\theta + G + H_c + {\cal J} ]^{n}(\widetilde\vartheta)\right\Vert_{[\varepsilon, T],\infty}  %
            \le~ \frac{L_{\varepsilon}^{n}T^{n}}{n!}\Vert\vartheta - \widetilde{\vartheta}\Vert_{[\varepsilon, T],\infty}
        \end{align*}
        and so it follows from \cite[Theorem 2.4]{Teschl-2016} that the operator $[ {\cal I}_\theta + G + H_c + {\cal J} ]$ has a unique fixed-point
        \[
            \vartheta^{\theta,c, \varepsilon}\in C^{0}({[\varepsilon,T]}).
        \]
        It follows from the uniqueness that the pointwise limit
        $
        \vartheta^{\theta,c}_t := \lim_{\epsilon \to 0}\vartheta_{t}^{\theta,c, \epsilon}
        $
        is well defined and satisfies
        \[
            [ {\cal I}_\theta + G + H_c + {\cal J}](\vartheta^{\theta, c})_t = \vartheta^{\theta, c}_t \quad \mbox{for all} \quad  t\in(0,T].
        \]
        Using the growth estimate \eqref{eq:GHJ_growth} and the dominated convergence theorem, we can uniquely extend $\vartheta^{\theta,c}$ to a continuous function on $[0,T]$.
        By construction $\vartheta^{\theta,c}$ is the unique fixed-point of $[ {\cal I}_\theta + G + H_c + {\cal J}]$ in $C^{0}([0,T])$, hence, the unique solution to the equation \eqref{eq:vartheta}. Thus, the unique solution to the equation \eqref{eq:mu-backwards} is given by $$\mu^{\theta,c} := \frac{\vartheta^{\theta,c}}{\eta}.$$

\end{proof}  
{\color{black}
\subsubsection{Existence and uniqueness of fixed-points}\label{sec:eu_fixedpoints}
                To establish the existence (and uniqueness) of a solution to our fixed-point equation, we next need solve the system \eqref{eq:bc_identities} for $(\theta, c)$, i.e., we need to prove that the function $\rho: [0,\infty) \times [0,\infty) \to \bR^2$, $(\theta, c) \mapsto (\rho_1(\theta, c), \rho_2(\theta, c))$ defined by
        \begin{equation}\label{eq:def_rho}
        \begin{split}
            \rho_1(\theta, c) &:=  \frac{\eta_T}{\widetilde\alpha^\delta_T}\theta - \mathbb E[\nu] + Q(c) - \int_0^T e^{\int_0^t\frac{A^\delta_r-\delta\kappa}{\eta_r}dr}\frac{A^\delta_t-\delta\kappa}{\eta_t}P\left(-\frac{1}{\alpha^\delta_t}\int_t^{T} e^{ -\int_0^{u}\frac{A^\delta_r}{\eta_r}\,dr    }\pi_u\vartheta^{\theta, c}_u\,du\right)\, dt,\\
            \rho_2(\theta, c) &:= c-\int_0^T h^\delta_s \pi_s \vartheta^{\theta,c}_s\,ds,
        \end{split}
        \end{equation}
        has a (unique) root, where $\vartheta^{\theta, c} \in C([0,T])$ is the solution to the equation \eqref{eq:vartheta}. 
        The following lemma proves the differentiability of $\vartheta^{\theta, c}$ with respect to the parameters $\theta$ and $c$ and establishes uniform bounds for the partial derivatives. The key observation is that the derivatives satisfy a Volterra equation with a generally discontinuous and unbounded kernel that depends non-linearly on the solution $\vartheta^{\theta, c}$.

   \begin{lemma}
        \label{lem:differentiation_bc}
        For all $\theta, c\ge 0$ let $\vartheta^{\theta, c} \in C([0,T])$ be the solution to \eqref{eq:vartheta}. Then the map $(\theta,c) \mapsto \vartheta^{\theta,c}_t$ is differentiable for all $t\in(0,T]$ and %
        \begin{align*}
            0 ~<~ \eta_T ~\le~ \frac{\partial \vartheta^{\theta,c}_t}{\partial \theta} ~\le~ K\frac{1}{h^\delta_t} \qquad \text{ and } \qquad - K\frac{1}{h^\delta_t} ~\le~ \frac{\partial \vartheta^{\theta,c}_t}{\partial c} ~\le~ 0 , \qquad t \in (0,T], \quad \theta, c \ge 0.
        \end{align*}
        Furthermore, let $\rho:[0,\infty)\times[0,\infty)\to \mathbb{R}^2$ be as defined in \eqref{eq:def_rho}. Then $\rho$ is differentiable and
        \begin{align*}
            \renewcommand*{\arraystretch}{2}
            D\rho_{(\theta, c)} =
            \begin{pmatrix}
                \frac{\eta_T}{\widetilde\alpha^\delta_T} +  \int_0^T \chi_s \frac{\partial \vartheta^{\theta, c}_s}{\partial \theta} \,ds & q(c) +  \int_0^T \chi_s \frac{\partial \vartheta^{\theta, c}_s}{\partial c} \,ds\\
                -\int_0^{T} h^\delta_s \pi_s \frac{\partial \vartheta^{\theta,c}_s}{\partial \theta} ds & 1 - \int_0^{T} h^\delta_s \pi_s \frac{\partial \vartheta^{\theta,c}_s}{\partial c} \,ds
            \end{pmatrix}
        \end{align*}
        where for all $t\in [0,T]$
        \begin{align*}
            0 ~\le~ \chi_t
            := \pi_te^{ -\int_0^{t}\frac{A^\delta_r}{\eta_r}\,dr    }\int_0^t \frac{e^{ 2\int_0^{s}\frac{A^\delta_r-\delta\kappa}{\eta_r}\,dr    }}{\eta_s}p\left(-\frac{1}{\alpha^\delta_s}\int_s^{T} e^{ -\int_0^{u}\frac{A^\delta_r}{\eta_r}\,dr    }\pi_u\vartheta^{\theta, c}_u\,du\right)\, ds ~\le~ Kp(0)h^\delta_t.
        \end{align*}

    \end{lemma}
    \begin{proof}
Our starting point is the equation \eqref{eq:vartheta}, which can be brought into the following form using an integration by parts argument (cf. \eqref{eq:proof_integration_by_parts}):
        \begin{align}
            \label{eq:vartheta_appendix}
            \begin{split}
                \vartheta_t =~& \theta\eta_T + \frac{1}{h^\delta_T}Q\left( c\right) - \frac{1}{h^\delta_t}Q\left( c - \int_t^T h^\delta_u \pi_u \vartheta_u du \right) + \int_{t}^{T}\frac{\dot{h}^\delta_s}{(h^\delta_s)^2}Q\left( c - \int_s^T h^\delta_u \pi_u \vartheta_u du \right)\, ds \\
                &- (A^\delta_t - \delta\kappa) P\left(-\frac{1}{\alpha^\delta_t}\int_t^{T} e^{ -\int_0^{u}\frac{A^\delta_r}{\eta_r}\,dr    }\pi_u\vartheta_u\,du\right) \\
                & - \int_{t}^{T} \frac{A_s^\delta(A^\delta_s-\delta\kappa)}{\eta_s} P\left(-\frac{1}{\alpha^\delta_s}\int_s^{T} e^{ -\int_0^{u}\frac{A^\delta_r}{\eta_r}\,dr    }\pi_u\vartheta_u\,du\right) \,ds\\
                &+  \int_t^{T}\lambda_s \left( \int_s^{T}\frac{1}{\eta_u}\vartheta_u \,du \right)\,ds - \int_t^T\delta\pi_s\vartheta_s\,ds.
            \end{split}
        \end{align}
        By the proof of Theorem~\ref{thm:existence}, we know that for any $\theta, c \ge 0$ there exists a unique solution $$\vartheta^{\theta, c} \in C^0([0,T])$$ to the above equation and the estimates \eqref{eq:H_estimate_lip} and \eqref{eq:GHJ_lip} show that the mapping
        \[
            (\theta, c) \mapsto \vartheta^{\theta, c}\vert_{[\varepsilon,T]} \in C^{0}([\varepsilon,T])
        \]
        is Lipschitz continuous, for any $\varepsilon > 0$.
        This allows us to apply the dominated convergence theorem to establish the differentiability w.r.t.~$\theta$ and to interchange differentiation and integration to obtain the following representation of the derivative:
        \begin{equation}
            \label{eq:mu-theta}
            \begin{split}
                \frac{\partial \vartheta^{\theta, c}_t}{\partial \theta}  =&~ \eta_T
                + \frac{1}{h^\delta_t}q\left( c - \int_t^T h^\delta_s \pi_s\vartheta^{\theta, c}_s   \,ds \right) \int_t^T h^\delta_s \pi_s \frac{\partial \vartheta^{\theta, c}_s}{\partial \theta} \,ds \\
                &- \int_{t}^{T} \frac{\dot{h}^\delta_s}{(h_s^\delta)^2} q\left( c - \int_s^T h^\delta_u \pi_u\vartheta^{\theta, c}_u \,du \right) \left( \int_s^T h^\delta_u \pi_u \frac{\partial \vartheta^{\theta, c}_u}{\partial \theta}    \,du \right)\, ds \\
                &  + (A^\delta_t-\delta\kappa) p\left(-\frac{1}{\alpha^\delta_t}\int_t^{T} e^{ -\int_0^{u}\frac{A^\delta_r}{\eta_r}\,dr    }\pi_u\vartheta^{\theta, c}_u\,du\right)\frac{1}{\alpha^\delta_t}\int_t^{T} e^{ -\int_0^s\frac{A^\delta_r}{\eta_r}\,dr    }\pi_s\frac{\partial \vartheta^{\theta, c}_s}{\partial \theta}\,ds \\
                & + \int_t^{T}\frac{A^\delta_s(A^\delta_s-\delta\kappa)}{\eta_s}p\left(-\frac{1}{\alpha^\delta_s}\int_s^{T} e^{ -\int_0^{u}\frac{A^\delta_r}{\eta_r}\,dr    }\pi_u\vartheta^{\theta, c}_u\,du\right)\frac{1}{\alpha^\delta_s}\left(\int_s^{T} e^{ -\int_0^u\frac{A^\delta_r}{\eta_r}\,dr    }\pi_u\frac{\partial\vartheta^{\theta, c}_u}{\partial \theta}\, du\right)\,ds \\
                &+  \int_t^{T}\lambda_s\left(\int_s^{T} \frac{1}{\eta_u}\frac{\partial \vartheta^{\theta, c}_u}{\partial \theta} \,du \right) \,ds - \int_t^T\delta\pi_s \frac{ \partial\vartheta^{\theta,c}_s }{\partial\theta}\,ds   \\
                =:&~ \eta_T + \int_t^{T} \Gamma(t,s) \frac{\partial \vartheta^{\theta, c}_s}{\partial \theta}  \,ds,
            \end{split}
        \end{equation}
        where the kernel $\Gamma$ admits the explicit representation
        \begin{align*}
            \Gamma(t, s)~=&~ \frac{h^\delta_s \pi_s}{h^\delta_t}C_t - h^\delta_s\pi_s\int_{t}^{s} \frac{\dot{h}^\delta_u}{(h^\delta_u)^2} C_u\, du +  \frac{1}{\eta_s}\int_t^{s}\lambda_u \,du  \\
            &+ D_t e^{ -\int_t^s\frac{A^\delta_r}{\eta_r}\,dr    }\pi_s
            + e^{ -\int_0^s\frac{A^\delta_r}{\eta_r}\,dr} \pi_s\int_t^{s}\frac{A^\delta_u(A^\delta_u-\delta\kappa)}{\eta_u}D_u\frac{1}{\alpha^\delta_u} \,du - \delta \pi_s
        \end{align*}
        for all $0 \le t \le s \le T$ with
        \[
            C_t := q\left( c - \int_t^T h^\delta_s \pi_s\vartheta^{\theta, c}_s \,ds \right), \quad
            D_t := p\left(-\frac{1}{\alpha^\delta_t}\int_t^{T} e^{ -\int_0^{u}\frac{A^\delta_r}{\eta_r}\,dr    }\pi_u\vartheta^{\theta, c}_u\,du\right).
        \]

        This shows that the derivative satisfies a Volterra integral equation, which suggests that the derivative can be bounded in terms of the kernel $\Gamma$. To this end, we first prove that $\Gamma$ is non-negative and then establish a growth condition on the kernel that carries over to our derivative function.

        \begin{itemize}
            \item {\bf Non-negativity of $\Gamma$.} The function $C$ is c\`agl\`ad and non-increasing, in particular of finite variation. Moreover, $h^\delta$ is continuous and increasing.
            Therefore, using the integration by parts formula for finite variation functions we can transform the first two terms of $\Gamma$ as follows
            \begin{align*}
                \frac{h^\delta_s \pi_s}{h_t}C_t - h^\delta_s\pi_s\int_{t}^{s} \frac{\dot{h}^\delta_u}{(h_u^\delta)^2}  C_u\, du  {-\delta\pi_s}  ~=~ \pi_s ( C_{s} {-\delta}) - h^\delta_s\pi_s\int_{[t,s)} \frac{1}{h^\delta_u} \,d C_u, \quad 0 < t \le s \le T.
            \end{align*}

            In the MFG $\delta = 0$ and hence the above term is non-negative. In the $N$-player game $\delta=\frac{1}{N}$. Let $x_N$ be the initial position of the largest seller, i.e.~the upper limit of the support of $\nu$. Since
            \[
                c-\int_s^Th^\delta_r\pi_r\vartheta^{\theta,c}_r\,dr  \leq c<Q^{-1}(\mathbb E[\nu])\leq Q^{-1}\left(\int_0^\infty q(x)\,dx\right)=x_N,
            \]
            we see that
            \[
                C_t=q\left( c - \int_t^T h^\delta_s \pi_s\vartheta^{\theta, c}_s \,ds \right)\geq \frac{1}{N},
            \]
            from which we again deduce non-negativity of the above term. All other terms in the definition of $\Gamma$ are non-negative as well.
            \item {\bf Growth bounds on $\Gamma$.} Using again that $h^\delta$ is increasing we see that
            \begin{align*}
                \left\vert \pi_s C_{s} - h^\delta_s\pi_s\int_{[t,s)} \frac{1}{h^\delta_u} \,d C_u \right\vert  ~\le~ K\left(1 + \frac{h^\delta_s}{h^\delta_t}\right), \qquad 0 < t \le s \le T.
            \end{align*}
            By \cite[Lemma~A.1]{FHH-2023} we have the following estimate
            \begin{align*}
                e^{ -\int_0^s\frac{A^\delta_r}{\eta_r}\,dr} \pi_s\int_t^{s}\frac{(A^\delta_u)^2}{\eta_u}D_u\frac{1}{\alpha^\delta_u} \,du ~\leq~ K(T-s)\int_t^s \frac{1}{(T-u)^2}\,du ~\leq~ K, \qquad 0 \le t \leq s \le T.
            \end{align*}
        \end{itemize}
        From the above estimates and the monotonicity of $h^\delta$ it follows that the modified kernel $$\widetilde{\Gamma}(t,s):= \Gamma(t,s)\frac{h^\delta_t}{h^\delta_s}, \qquad 0 \le t \le s \le T$$ is non-negative and bounded.
        Results established in \cite{beesack1969comparison} show that
        \[
            y(t) := h^\delta_t\frac{\partial{\vartheta^{\theta, c}_t}}{\partial \theta}, \quad t\in[0,T]
        \]
        is the unique and bounded solution to the Volterra integral equation
        \begin{align*}
            y(t) =&~ h^\delta_t\eta_T + \int_t^{T} \widetilde\Gamma(t,s) y(s)  \,ds, \quad t\in[0,T].
        \end{align*}
        In particular, it holds
        \begin{align*}
            \eta_T ~\le~ \frac{\partial{\vartheta^{\theta, c}_t}}{\partial \theta} ~\le~ K\frac{1}{h^\delta_t}, \qquad 0 < t \le T.
        \end{align*}

        An analogous argument establishes the differentiability of the function $\vartheta^{\theta, c}_t$ with respect to the parameter $c$ and shows that
        \[
            y(t) = h^\delta_t\frac{\partial{\vartheta^{\theta, c}_t}}{\partial c}, \quad t\in[0,T]
        \]
        uniquely solves the integral equation
        \begin{align*}
            y(t) =&~ z(t) + \int_t^{T} \widetilde{\Gamma}(t,s) y(s)  \,ds, \quad t\in[0,T],
        \end{align*}
        where %
        \begin{align*}
            z(t) := h^\delta_t\left( \frac{1}{h^\delta_T} q\left( c\right) - \frac{1}{h^\delta_t}C_t + \int_{t}^{T}\frac{\dot{h}^\delta_s}{(h^\delta_s)^2}C_s\, ds \right)~=~ h^\delta_t\int_{[t, T)} \frac{1}{h^\delta_u} \,d C_u.
        \end{align*}
        As before we see from the right-hand side that $z$ is non-positive and bounded.
        Hence, it follows that
        \begin{align*}
            -K\frac{1}{h^\delta_t} ~\le~ \frac{\partial{\vartheta^{\theta, c}_t}}{\partial c} ~\le~ z(t)\frac{1}{h^\delta_t} \le 0, \qquad 0 < t \le T.
        \end{align*}

        To prove that $\rho$ is differentiable we have to once again justify that differentiation w.r.t.~$\theta$ (resp.~$c$) is interchangeable with the integrals in the definition of $\rho$.

        To this end, we notice that the above bounds for $\left|\frac{\partial{\vartheta^{\theta, c}}}{\partial \theta}\right|$ and $\left|\frac{\partial{\vartheta^{\theta, c}}}{\partial c}\right|$ hold uniformly in $\theta, c \ge 0$.
        Thus it suffices to show that these bounds provide integrable majorants.
        For $\rho_2$ this follows from the presence of the factor $h^\delta$ in the integrand. Regarding $\rho_1$ we recall that by \cite[Lemma~2.6]{FHH-2023} we have that $h^\delta_t \ge K t$ for all $t\ge 0$ and thus,
        \begin{align*}
            &\int_0^T \frac{A^\delta_t-\delta\kappa}{\alpha^\delta_t \eta_t}p\left(-\frac{1}{\alpha^\delta_t}\int_t^{T} e^{ -\int_0^{u}\frac{A^\delta_r}{\eta_r}\,dr    }\pi_u\vartheta^{\theta, c}_u\,du\right)  e^{  \int_0^t \frac{ A^\delta_r-\delta\kappa }{\eta_r}\,dr  }  \left(\int_t^{T}e^{ -\int_0^{s}\frac{A^\delta_r}{\eta_r}\,dr    }\pi_s \frac{1}{h_s} \,ds\right)\, dt \\
            &\le K \int_0^T(A^\delta_t - \delta\kappa) \left(\int_t^{T}\frac{1}{s} ds \right)\,dt \\
            &\le K \int_0^T \frac{1}{T-t} (\log(T) - \log(t)) \,dt\\
            &< \infty.
        \end{align*}

        A straightforward computation using Fubini's theorem now shows that
        \begin{align*}
            \frac{\partial \rho_1}{\partial \theta}(\theta, c)
            =&~ \frac{\eta_T}{\widetilde\alpha^\delta_T} +  \int_0^T \frac{A^\delta_t-\delta\kappa}{\alpha^\delta_t \eta_t}p\left(-\frac{1}{\alpha^\delta_t}\int_t^{T} e^{ -\int_0^{u}\frac{A^\delta_r}{\eta_r}\,dr    }\pi_u\vartheta^{\theta, c}_u\,du\right)   e^{\int_0^t\frac{ A^\delta_r-\delta\kappa  }{\eta_r}\,dr}    \\
            &~ \qquad \qquad    \times \left(\int_t^{T}e^{ -\int_0^{s}\frac{A^\delta_r}{\eta_r}\,dr    }\pi_s \frac{\partial \vartheta^{\theta,c}_s}{\partial \theta} \,ds\right)\, dt \\
            =&~ \frac{\eta_T}{\widetilde\alpha^\delta_T} +  \int_0^T \frac{\partial \vartheta^{\theta,c}_t}{\partial \theta} \pi_t e^{ -\int_0^{t}\frac{A^\delta_r}{\eta_r}\,dr    }\left(\int_0^t e^{ \int_0^{s}\frac{A^\delta_r-\delta\kappa}{\eta_r}\,dr    }\frac{A^\delta_s-\delta\kappa}{\alpha^\delta_s \eta_s} \right. \\
            &~ \left. \qquad \qquad \times p\left(-\frac{1}{\alpha^\delta_s}\int_s^{T} e^{ -\int_0^{u}\frac{A^\delta_r}{\eta_r}\,dr    }\pi_u\vartheta^{\theta, c}_u\,du\right)\, ds\right)\,dt \\
            =&~ \frac{\eta_T}{\widetilde\alpha^\delta_T} +  \int_0^T \frac{\partial \vartheta^{\theta,c}_t}{\partial \theta}\chi^{\theta, c}_t \,dt.
        \end{align*}

        The derivation of the remaining partial derivatives is analogous.
        The fact that $\chi^{\theta, c}$ is non-negative follows from its definition, while its upper bound follows from the definition of $h^\delta$.
    \end{proof}
}
    We are now ready to prove the existence (and uniqueness) of the roots of the map~$\rho$, and thus the existence (and uniqueness) of fixed points of~$F$.
    
    \begin{theorem}
        \label{thm:existence2}
        Assume that $\E[\nu] > 0$.
        \begin{enumerate}
            \item[(i)] There exists $\theta, c \ge 0$ such that the identities \eqref{eq:bc_identities} hold.
            In particular, $\mu^{\theta,c}$ is a fixed-point of $F$.

            \item[(ii)] If $p(0)$ is small enough, then there exists a unique pair $(\theta,c)$ such that \eqref{eq:bc_identities} holds. In this case, $\mu^{\theta,c}$ is the unique fixed-point of $F$.
        \end{enumerate}
    \end{theorem}
\begin{proof}
    By Theorem~\ref{thm:existence}, there exists a unique solution $\mu^{\theta, c} = \eta^{-1} \vartheta^{\theta, c}$ to~\eqref{eq:mu-backwards} for all $c, \theta \ge 0$. Hence, to prove the first claim, it suffices to show that there exist $c, \theta \ge 0$ such that~\eqref{eq:bc_identities} holds. To prove the second claim, we must show that such a pair $(\theta, c)$ is unique. This is equivalent to showing that the map $\rho$, defined in~\eqref{eq:def_rho}, has a (unique) root.
    To this end, we first notice that any such root necessarily satisfies\footnote{Note that  $\lim_{x\to\infty} Q(x) = \int_0^\infty q(x)\,dx > \E[\nu]$. Furthermore, $Q$ is increasing and strictly increasing on the interval $Q^{-1}(\bR)$, hence,  $Q^{-1}(\E[\nu])$ is well defined.}
        \[
            c < Q^{-1}(\E[\nu])).
        \]

        We now proceed in two steps. We first prove that for any $c\in [0, Q^{-1}(\E[\nu]))$ there exists a unique $\theta(c) \in (0,\mathbb{E}[\nu]\frac{\widetilde\alpha^\delta_T}{\eta_T})$ such that
        \begin{align}
            \label{eq:implicit_def_bc}
            \rho_{1}(\theta(c),c) = 0.
        \end{align}
        In fact, from Lemma~\ref{lem:differentiation_bc} it follows that the mapping $\theta \mapsto \rho_1(\theta, c)$ is strictly increasing and continuous.
        Thus, it suffices to show that this map changes its sign on $[0,\mathbb{E}[\nu]\frac{\widetilde\alpha^\delta_T}{\eta_T}]$.
        Choosing $\theta = 0$ we have $\vartheta_T^{0, c} = 0$, hence $\vartheta^{0, c} \equiv 0$, and therefore $\rho_1(0, c) = Q(c) - \mathbb{E}[\nu] < 0$.
        On the other hand, choosing $\theta = \mathbb{E}[\nu]\frac{\widetilde\alpha^\delta_T}{\eta_T}$ we have $\vartheta_T^{\theta,c} >  0$, hence $\vartheta^{\theta,c} > 0$, and therefore
        \begin{align*}
            \rho_1( \mathbb{E}[\nu]\frac{\widetilde\alpha^\delta_T}{\eta_T}, c) = Q(c) - \int_0^T e^{\int_0^t\frac{A^\delta_r-\delta\kappa}{\eta_r}dr}\frac{A^\delta_t-\delta\kappa}{\eta_t}P\left(-\frac{1}{\alpha^\delta_t}\int_t^{T} e^{ -\int_0^{u}\frac{A^\delta_r}{\eta_r}\,dr    }\pi_u\vartheta^{\theta, c}_u\,du\right) dt  > 0.
        \end{align*}

        It remains to show that the function $c \mapsto \rho_{2}(\theta(c),c)$
        has a root. By Lemma \ref{lem:differentiation_bc} and the implicit function theorem it follows that the function $c \mapsto \theta(c)$ is continuous. Hence, by Lemma~\ref{lem:differentiation_bc}, the function $c \mapsto \rho_{2}(\theta(c),c)$ is also continuous, and it suffices to show that it changes its sign on the interval $(0, Q^{-1}(\E[\nu]))$.

        Choosing $c = 0$ and recalling that $\theta(0) > 0$, hence $\vartheta^{\theta(0), 0} > 0$, we see that
        \begin{align*}
            \rho_{2}(\theta(0),0) = - \int_0^T h^\delta_s \pi_s \vartheta^{\theta(0), 0}_s ds < 0.
        \end{align*}
        On the other hand, if $c \to  Q^{-1}(\E[\nu])$, then $\theta(c) \to 0$, hence $\Vert \vartheta^{\theta(c), c} \Vert_{\infty}\to 0$, and so
        \begin{align*}
            \lim_{c \to Q^{-1}(\E[\nu])}\left(c - \int_0^T h^\delta_s \pi_s \vartheta^{\theta(0), c}_s \,ds\right) > 0.
        \end{align*}

        By the implicit function theorem and Lemma~\ref{lem:differentiation_bc} below we have that
        \begin{align*}
            \frac{\partial \theta(c)}{\partial c}
            =- \frac{\partial \rho_1}{\partial c} \left(\frac{\partial \rho_1}{\partial \theta}\right)^{-1}
            = - \frac{q(c)+\int_0^T \chi^{\theta, c}_t \frac{\partial \vartheta^{\theta, c}_t}{\partial c} dt}{ \frac{\eta_T}{\widetilde\alpha^\delta_T} +  \int_0^T \chi^{\theta, c}_t \frac{\partial \vartheta^{\theta, c}_t}{\partial \theta} dt}
            \le K p(0), \quad c \ge 0.
        \end{align*}
        Using once again the uniform estimates from Lemma~\ref{lem:differentiation_bc}, we see that
        \begin{align*}
            \frac{d}{dc}\rho_2(\theta(c), c)
            =&~ 1 -\int_0^{T} h^\delta_s \pi_s \left\{ \frac{\partial \vartheta^{\theta,c}_s}{\partial \theta}\frac{\partial \theta(c)}{\partial c} +  \frac{\partial \vartheta^{\theta,c}_s}{\partial c}\right\}ds \\
            \ge&~ 1 - K p(0).
        \end{align*}
        For small enough $p(0)$, the function $c \mapsto \rho_2(\theta(c), c)$ is strictly increasing, and therefore, the root is unique.
    \end{proof}

    \subsection{Existence and uniqueness of equilibria}\label{sec:proof_main}

    With our fixed-point results in hand, we are now ready to establish our existence and uniqueness of equilibrium results.
    \begin{proof}[Proof of Theorem~\ref{thm:main}]
     We will first assume that $\E[\nu] = \int_{\mathbb{R}} x\nu(d x) > 0$. The case $\E[\nu] < 0$ follows by symmetry and the case  $\E[\nu] = 0$ is specifically treated in (iv).
        \begin{enumerate}[label=(\roman*)]
            \item Since In view of Assumption~\ref{ass:initial_distribution} it follows from Theorem~\ref{thm:existence2}.(i) that there exists a fixed-point $\mu^\ast \in C([0,T])$ of $F$, i.e., 
            $$ \mu^\ast_t = \int_{\bR}\xi^{*,\delta, x, \mu^\ast}_t \nu(dx) $$
            where we recall that the strategies $\xi^{*,\delta, x, \mu^\ast}$ are defined in \eqref{xi*3} for $x \ge 0$ and \eqref{xi-sigma} for $x < 0$.
            By Proposition~\ref{prop:fix-point-integral} the fixed point $\mu^\ast$ is a solution to the integral equation~\eqref{eq:mu-first-integral} with terminal condition~\eqref{eq:terminal_condtion} and satisfies $\mu_t > 0$ for all $t \ge 0$. By Lemma~\ref{lem:apriori}, the mapping $t \mapsto \mu_t \eta_t$ is  decreasing, which verifies that Assumption~\ref{ass:mu_sign} is satisfied for $\mu = \mu^\ast$.
            Furthermore, by the assumption on $\lambda$ (resp. $\lambda\eta$) it follows from Proposition~\ref{prop:sufficient_conditions} that Assumption~\ref{ass:ceof_relation} is satisfied as well.
            It thus follows from our verification results stated in Theorem~\ref{prop:sellers} and \ref{thm:veri-N} that the strategies $\xi^{\ast, \delta, x, \mu^\ast}$ form a Nash equilibrium, both for the $N$-player ($\delta = \frac 1 N$) and the mean-field ($\delta=0$) game.%

            \item Given that $p(0) = \nu((-\infty,0])$ is sufficiently small, the uniqueness of the fixed point $\mu^\ast$ in Theorem~\ref{thm:existence2}(ii), together with the uniqueness of the best response established in Theorems~\ref{thm:veri-N} and~\ref{prop:sellers}, implies that the equilibrium is also unique within the class of equilibria whose aggregate rate $\mu$ satisfies Assumption~\ref{ass:mu_sign}. 
            \footnote{The sign condition on $\mu$ can be relaxed to a mere continuity requirement when using a dynamic subgame perfection argument, as in \cite[Theorem~4.1]{FHH-2023}. In contrast, the monotonicity condition on $\mu\eta$—which is new compared to \cite{FHH-2023}—does not appear to admit a straightforward relaxation.}%

            \item %

            We may assume without loss of generality that $\nu$ satisfies the condition of a seller dominated market.
            It then follows from the convergence assumptions on the sequence $(\nu^N)_{N\in\mathbb{N}}$ that $$\lim_{n\to\infty}\E[\nu^N] = \E[\nu].$$ In particular, for $N_0 \in\mathbb{N}$ large enough it also holds $\mathbb{E}[\nu^N] > 0$ for all $N\ge N_0$.
            Further, since $\nu$ has  density, weak convergence implies the uniform convergence of the cumulative distribution functions, so in particular
            for ${p(x) = \nu((-\infty, x])}$, $p^N(x) = \nu^{N}((-\infty, x])$, $q(x) = \nu([x,\infty))$ and $q^N(x) = \nu^{N}([x,\infty))$ we have
    \begin{equation}\label{ass:convergence}
             \lim_{N\to\infty}\sup_{x \ge 0} \Big( \big\vert p^N(x) -p(x) \big\vert +  \big\vert q^N(x) - q(x) \big\vert \Big) = 0.
            \end{equation}
            So in particular by choosing $N_0$ possibly larger we have from  $p(0) = \nu((-\infty,0])<\nu_0$ that also $p^N(0) = \nu^N((-\infty,0])<\nu_0$ for all $N\ge N_0$.

            For $N \ge N_0$ it follows from the uniqueness result in (ii) that the aggregate trading rate of the $N$-player game $\mu^N$ solves the integral equation \eqref{eq:mu-first-integral} with $\delta =\frac{1}{N}$ and the terminal condition \eqref{eq:terminal_condtion}.
            Reformulating this equation with the substitutions $\vartheta^N := \mu^N \eta$ and $\pi := \frac{\kappa}{\eta}$ we obtain
               \begin{equation}\label{eq:etaN_convergence}
            \begin{split}
                \vartheta^N_t  =&~ \vartheta^N_T +\int_t^T\lambda_s\int_s^{T}\frac{1}{\eta_u}\vartheta^N_u \,du +\int_{t}^{T} q^N\left( \phi_{\mu^N}(s) \right) {\pi_s} \vartheta^N_s\, ds \\
                &~+  \int_t^{T}{\lambda_s} \  \ell^N\left(-\psi_{\mu^N}(s)\right)   \,ds 
                +\int_t^T\left\{   p^N\left(-\psi_{\mu^N}(s)\right) - \frac{1}{N} \right\}\pi_s\vartheta^N_s\,ds.
            \end{split}
        \end{equation}
        We first verify the sequence $( \vartheta^N )_{N\ge N_0}$ admits a limit in $C^0([0,T])$. 
        Note that for terminal value it follows from its expression in \eqref{eq:terminal_condtion}, Assumption~\ref{ass:convergence} and \cite[Lemma A.4]{FHH-2023} that
        $$0<\vartheta^N_T\leq\sup_{N\ge N_0} \widetilde\alpha^{1/N}_T\mathbb E[\nu^N]<\infty.$$ 
        The last inequality above is implied by the convergence condition on $\nu^N$.
        Next, following the growth bound in the proof of Theorem~\ref{thm:existence} Step~2, stressing that the constants $K$ do not depend on $\delta \in [0,1]$, we obtain that %
        $\sup_{N\ge N_0}\|\vartheta^N \|_\infty<\infty$. 
        Similarly, using the boundedness of $q^N$, $p^N$ and the following estimate for the unbounded non-linear function $\ell^N$,
        $$\ell^N( -\psi_{\mu^N}(s))\leq \frac{2}{\alpha_s^{1/N}}\int_s^T e^{-\int_0^r \frac{A_u^{1/N}}{\eta_u} \,du}\kappa \mu^N_r\,dr \leq K\sup_N\| \vartheta^N   \|_\infty <\infty,$$
        we can bound the Lipschitz constant $|\vartheta^N_t-\vartheta^N_s| \leq K\sup_{N\ge N_0}\|\vartheta^N\|_\infty |t-s|  $, uniformly over all $s,t\in [0,T]$.
        Thus, Arzel\`a-Ascoli theorem yields a limit of $\vartheta^N$ in $C^0([0,T])$, denoted by $\vartheta$. 
        
        Finally letting $N\rightarrow\infty$ in \eqref{eq:etaN_convergence}, and using the previous growth bounds for the application of the dominated convergence theorem, we see that the limit $\vartheta$ satisfies \eqref{eq:vartheta_appendix} with $\delta=0$, which readily establishes the convergence to the unique mean field equilibrium.
            
    \item In the case $\mathbb{E}[\nu] = 0$, we set $\mu^\circ \equiv 0$ and observe that $\phi_{\mu^\circ} \equiv \psi_{\mu^\circ} \equiv 0$.  
    One can then verify directly that $\mu^\circ$ solves the integral equation~\eqref{eq:mu-first-integral} with the  terminal condition~\eqref{eq:terminal_condtion}, which is trivial in this case.
    Since Assumptions~\ref{ass:ceof_relation} and~\ref{ass:mu_sign} are also trivially satisfied, we can verify—just as in case (i)—that $(\xi^{\ast, \frac{1}{N}, x_i, \mu^\circ})_{i = 1, \dots, N}$ and $(\xi^{\ast, 0, x, \mu^\circ})_{x \in \mathbb{R}}$ constitute equilibria for the $N$-player and mean-field games, respectively.
    Note that since $\phi_{\mu^\circ} \equiv \psi_{\mu^\circ} \equiv 0$, it follows that the optimal entry times satisfy $\sigma_{\mu^\circ} \equiv 0$ (see~\eqref{eqn:sigma}) and the optimal exit times satisfy $\tau_{\mu^\circ} \equiv T$ (see~\eqref{eqn:tau}). Hence, all players on both sides of the market enter immediately and exit at the terminal time $T$.
        \end{enumerate}
    \end{proof}

    \begin{remark}\label{rem:focal-equi}
    \begin{itemize}
        \item[(i)] In general, we cannot rule out the existence of an equilibrium rate that is not monotone towards the end of the trading period.
     In case $\E[\nu] = 0$, for instance, the unique equilibrium within the previously mentioned class is given by $\mu \equiv 0$.
     However, we cannot rule out the existence of an equilibrium rate that changes its sign infinitely often.
     Such equilibria are much less ``focal'' and are therefore not relevant.
     
    \item[(ii)] Note that the convergence result in \cite{FHH-2023} depends only on the initial distribution restricted to the positive half-space, whereas the convergence result in this paper allows for initial distributions defined over the entire space. 
    \end{itemize}
    \end{remark}

    \section{Examples}\label{sec:examples}

    In what follows, we present numerical examples to illustrate how our constraint of the trading direction affects equilibrium trading in both the mean field and the $N$-player games.
    Therefore, we contrast our results with the equilibrium obtained under the market dropout constraint studied in \cite{FHH-2023} and the unconstrained case studied in \cite{FGHP-2018}.
    For simplicity we consider constant cost parameters; precisely we set
    \[
        \eta \equiv 5, \quad \kappa \equiv 10, \quad \mbox{and} \quad \lambda \equiv 5.
    \]

    \begin{figure}[h!]
        \begin{center}
             \includegraphics[scale=.64, trim=0 0 0 0,clip]{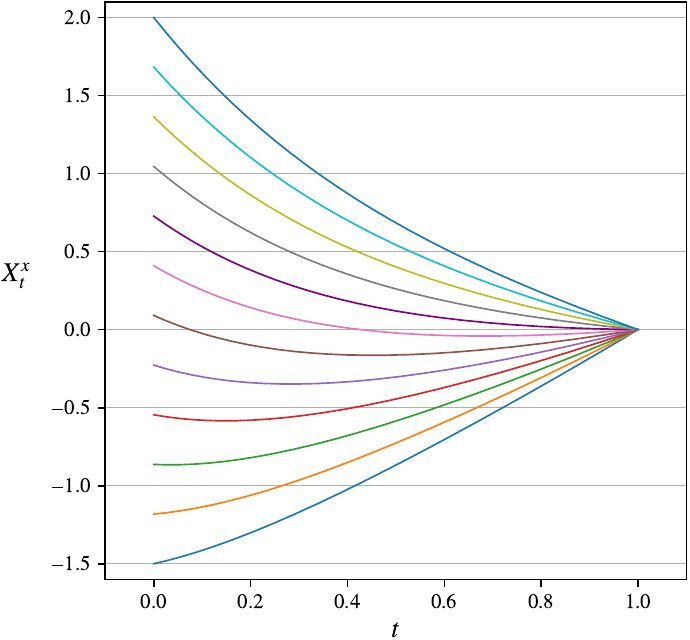}
            \hfill{}
            \includegraphics[scale=.64, trim=14 0 0 0,clip]{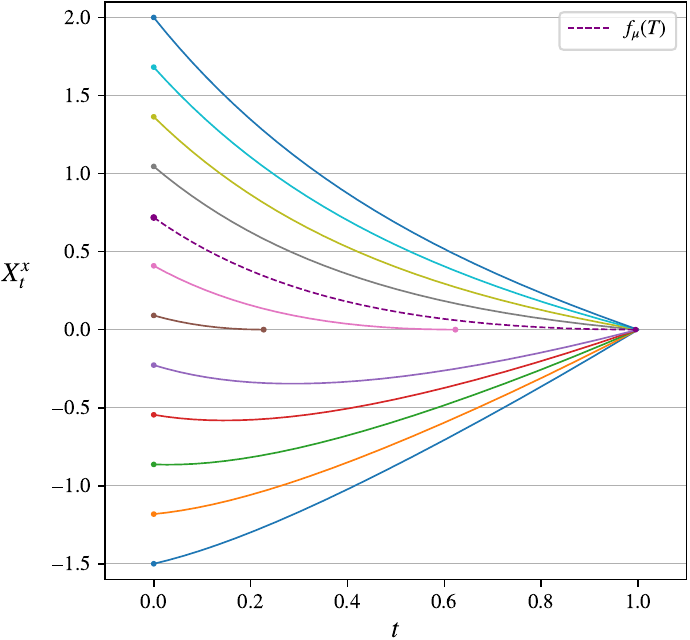}
            \includegraphics[scale=.64]{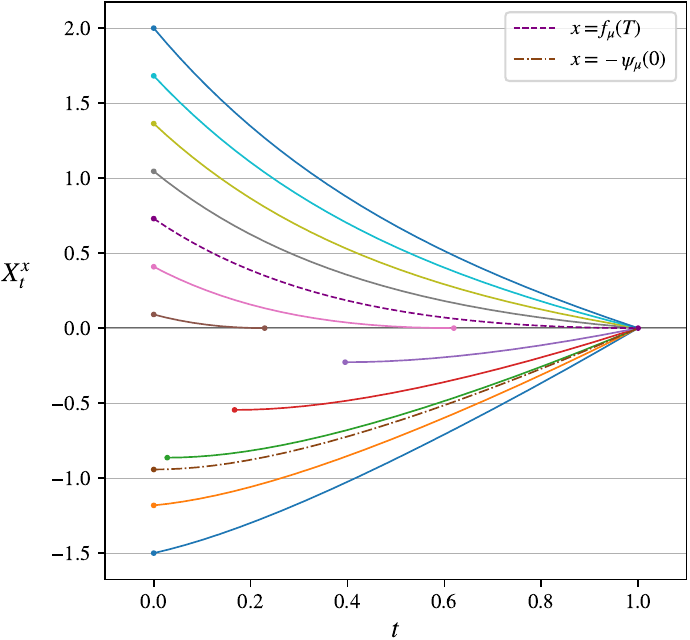}\hfill{}\includegraphics[scale=.64, trim=0 0 0 0,clip]{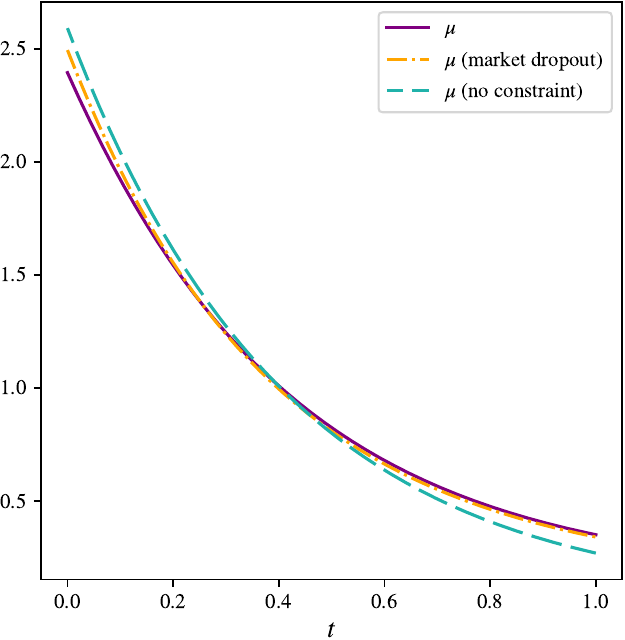}
        \end{center}
        \caption{Evolution of the equilibrium state processes \emph{without constraint} (top-left), \emph{with market drop-out} (top-right) and \emph{with trading constraint} (bottom-left) for several representative players.
        We have highlighted the moments of market entry and drop-out and the initial position $x = -\psi_\mu(0)$ and $x = \phi_\mu(T)$ (which represent the smallest initial positions for which one has immediate entry, respectively no early exit).
        In the bottom-right plot we compare the evolution of the mean trading rates of all three scenarios.}\label{fig:mfg}
    \end{figure}
    To approximate the mean-field equilibrium numerically, we first apply a standard numerical solver to integrate the backward equation \eqref{eq:mu-backwards} across varying values of $(\theta, c)$. Subsequently, we employ a standard root-finding procedure to identify a pair $(\theta, c) \in (0,\mathbb{E}[\nu]\frac{\widetilde\alpha^\delta_T}{\eta_T}) \times [0, Q^{-1}(\E[\nu]))$ that satisfies the equation \eqref{eq:bc_identities}, effectively finding a root of the function $\rho: \mathbb{R}^2 \to \mathbb{R}^2$ as defined in \eqref{eq:def_rho}.

    \begin{remark}
        For the benchmark case of constant coefficients, a closed-form solution $A^{\delta}$ for the Riccati equation \eqref{eq:AB+} is available (see  \cite{FGHP-2018}), which substantially simplifies the numerical implementation.
    \end{remark}

    We first consider an MFG with exponentially distributed initial positions on both sides of the market, setting
    \[
        q(x) = 0.8 \cdot e^{- \frac{2}{3} x} \quad \mbox{and} \quad  p(x) = 0.2 \cdot e^{x}.
    \]
    This results in an average initial position of $\E[\nu] = 1$, that is, in a seller dominated market. Figure~\ref{fig:mfg} presents the evolution of the equilibrium state processes for all three scenarios: no trading constraints, drop-out constraints and with trading constraints and several representative players.

    In models without constraints (top-left) we see that players on both sides of the market change the direction of trading for small initial positions. In a seller dominated market buyers can take advantage of favorable price trends. Hence it is beneficial for both sellers and buyers with small initial positions to (further) sell the asset and then buying it back at favorable prices.

    Under the market drop-out constraint (top-right), sellers do not change the direction of trading but may exit the market early. On the buyer side, however, we continue to observe players that initially use an opposite trading direction to benefit from the overall market trend. Our trading constraints avoid such effects. In a model with trading constraints (bottom-left) we see that it is beneficial for buyers with small short position to enter the market at later time points.

    Figure~\ref{fig:mfg} (bottom-right) presents a comparison of the average equilibrium trading rate across all three scenarios. With the parameters selected, the deviation in tradings rates is small. This is intuitive as only traders with small initial positions, hence with comparably small impact on the market dynamics, enter the market late, respectively, exit the market early.

    At the same time, we observe that our trading constraint slightly amplifies the effect previously observed under the market drop-out constraint, namely a slower initial aggregate liquidation, followed by an acceleration in aggregate liquidation halfway through the trading period, in comparison to the model without constraints.

    This dynamics can be intuitively understood by considering the impact of small buyers who, in the absence of trading constraints, would initially increase their short positions, thereby generating additional selling pressure. Under our trading constraint, these buyers are restricted to hold their position initially. Thus, there is initially no contribution to the aggregate trading rate, which ultimately also results in a higher aggregate trading rate later in the game as there is no need for buying back initially sold stocks.

    \begin{figure}[h!]
        \begin{center}
            \hspace{.55cm} \includegraphics[scale=.64, trim=0 0  0 0,clip]{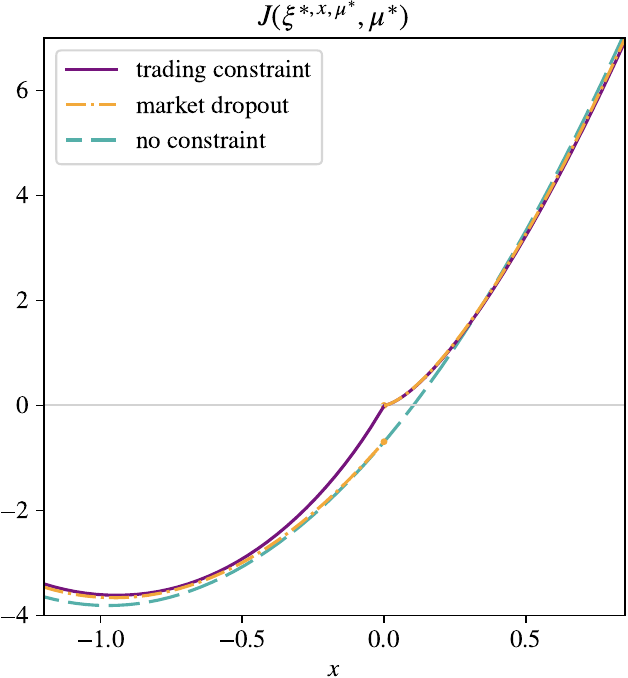}\hfill{}
             \includegraphics[scale=.64, trim=0 -.49cm 0 0,clip]{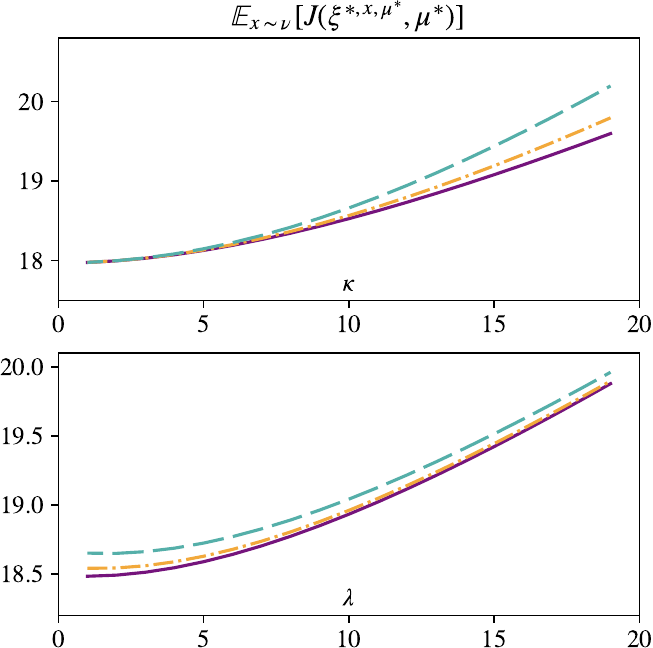}            \hfill{}
        \end{center}
        \caption{\textit{left}: Liquidation (acquisition) cost for individual players in the equilibrium plotted as a function of the initial position for all three scenarios. \textit{right}: Dependence of overall averaged costs in the equilibrium on the permanent price impact parameter $\kappa$ (top) and risk aversion parameter $\lambda$ (bottom).}\label{fig:costs_sensitivity}
    \end{figure}
    
        \begin{figure}[h!]
        \begin{center}
            \includegraphics[scale=.64, trim=0 0 0 0,clip]{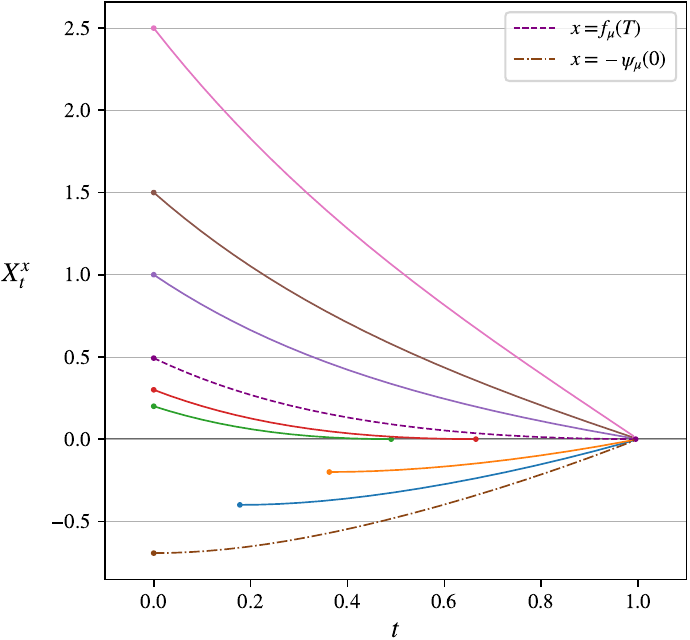}
            \hfill{}
            \includegraphics[scale=.64, trim=0 0 0 0,clip]{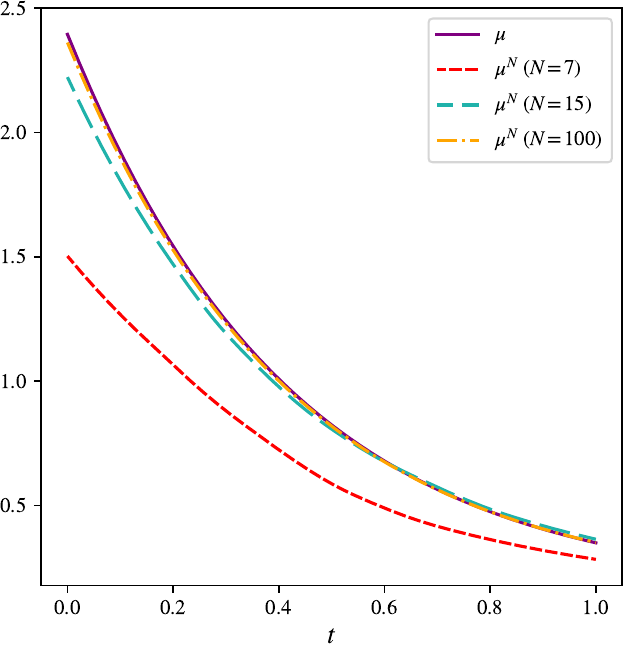}
        \end{center}
        \caption{\textit{left}: Evolution of states of $N=7$ players in the Nash equilibrium.
        \textit{right}: Comparison of the mean trading rate in the equilibrium of the \emph{MFG} (solid line) and the $N$-player game (dashed/dotted line) for $N=7, 15$ and $100$.}\label{fig:n-player}
    \end{figure}

    We note that in all three settings the relation $\int_0^T\mu_t\, dt = \mathbb{E}[\nu]$ holds. Hence, the price at the terminal time is the same in all settings. Nevertheless the distribution of liquidation (acquisition) costs across players differs significantly. In Figure~\ref{fig:costs_sensitivity} (left) we present the total costs of individual players as a function of their initial position for all three scenarios. The effect of the different constraints is clearly visible around $x=0$, where in presence of the trading constraint the profits of small buyers (i.e., their negative costs) are reduced by exclusion of short selling strategies. Furthermore, the asymmetry in strategies on both sides (late entry for buyers and early exit for sellers) is visible by the discontinuity of the cost functions derivative around at $x=0$. 
    Under the market dropout constraint, the asymmetry (short selling for buyers and early exit for sellers) leads to a discontinuity of the cost function itself.
    
    In Figure~\ref{fig:costs_sensitivity} (right) we see important differences in the dependency of average costs in equilibrium $\mathbb{E}_{x\sim\nu}[J(\xi^{\ast, x, \mu^\ast}, \mu^{\ast})]$ on the model parameters $\kappa$ and $\lambda$. 
It can be observed that as the permanent price impact factor $\kappa$ increases,
the average costs become increasingly lower in the presence of constraints, particularly under the trading
constraint.
This suggests that constraints are socially desirable in markets where the permanent impact is large.
    When $\kappa$ increases, stronger negative price trends from trades of the dominating seller side of the market offers more arbitrage opportunity for buyers by using short selling strategies, which even amplify the price trend initially. Removing short selling strategies from the market, the trading constraint reduces overall trading and hence overall costs.
    We also observe that the difference in average costs is largest when $\lambda$ is small.
    Increasing $\lambda$, short selling strategies become  less attractive due to high inventory costs, thereby gradually removing the associated arbitrage and its influence on the average costs. 
    
    Figure~\ref{fig:n-player} (left) illustrates the resolution of an $N$-player game with seven players.
    In this setup, the two buyers both have a starting position above $x = -\psi_\mu(0)$ and thus delay their market entry.
    The two sellers with initial positions below the threshold $x = \phi_\mu(0)$ fully liquidate their trades before reaching the terminal time. Figure~\ref{fig:n-player} (right) shows the aggregate rates for several numbers of players. This simulation further supports the observation of a fast convergence to the MFG equilibrium.

    \section{Conclusion}\label{sec:conclusion}

    We established existence and uniqueness of equilibrium results in multi-player and mean-field games of portfolio liquidation under a ``no change of trading condition''. We proved that the games are equivalent to games of timing where buyers and sellers need to determine the equilibrium market entry and exit times. Several avenues are open for future research.

    First, we worked under the assumption of deterministic market parameters. Although it would clearly be desirable to allow for stochastic parameters, it is unclear to us how to extend our model to the stochastic case. In our setting, entry and exit times are deterministic. In stochastic settings these times were stopping times and our equilibrium analysis would require fixed-point results for stopping times, which is challenging for many reasons. Most importantly we are unaware of any topology on the set of stopping times that would guarantee that (i) the set of stopping times is compact, and at the same time that (ii) our response functions mapping anticipated entry and exit times into actual entry and exit times would be continuous.

    A second limitation that one would like to overcome is our assumption that all players share the same liquidation time. In \cite{FHH-2023} we illustrate how our current results could be used to solve finite player games with heterogeneous trading horizons, but the approach is tedious and not very elegant. It would be desirable to develop a general game-theoretic framework that allows for heterogeneous liquidation times.

    \bibliographystyle{plain}

\begin{thebibliography}{10}

\bibitem{DuTa-2021b}
R.~A\"id, R.~Dumitrescu, and P.~Tankov.
\newblock The entry and exit game in the electricity markets: A mean-field game
  approach.
\newblock {\em Journal of Dynamics and Games}, 8(4):331--358, 2021.

\bibitem{AC-2001}
R.~Almgren and N.~Chriss.
\newblock Optimal execution of portfolio transactions.
\newblock {\em Journal of Risk}, 3:5--40, 2001.

\bibitem{AJK-2014}
S.~Ankirchner, M.~Jeanblanc, and T.~Kruse.
\newblock B{SDE}s with singular terminal condition and a control problem with
  constraints.
\newblock {\em SIAM Journal on Control and Optimization}, 52(2):893--913, 2014.

\bibitem{bank:voss:16}
P.~{Bank} and M.~{Vo{\ss}}.
\newblock {Linear quadratic stochastic control problems with singular
  stochastic terminal constraint}.
\newblock {\em SIAM Journal on Control and Optimization}, 56(2):672--669, 2018.

\bibitem{beesack1969comparison}
P.~R. Beesack.
\newblock Comparison theorems and integral inequalities for {Volterra} integral
  equations.
\newblock {\em Proc. Am. Math. Soc.}, 20:61--66, 1969.

\bibitem{DuTa-2020}
G.~Bouveret, R.~Dumitrescu, and P.~Tankov.
\newblock Mean-field games of optimal stopping: A relaxed solution approach.
\newblock {\em SIAM Journal on Control and Optimization}, 58(4):1795--1821,
  2020.

\bibitem{campi2021}
L.~Campi and M.~Burzoni.
\newblock Mean field games with absorption and common noise with a model of
  bank run.
\newblock {\em Stochastic Processes and Their Applications}, 164:206--241,
  2023.

\bibitem{campi2018}
L.~Campi and M.~Fischer.
\newblock {$N$}-player games and mean-field games with absorption.
\newblock {\em Annals of Applied Probability}, 28(4):2188--2242, 2018.

\bibitem{campi2020}
L.~Campi, M.~Ghio, and G.~Livieri.
\newblock {$N$}-player games and mean-field games with smooth dependence on
  past absorptions.
\newblock {\em Annales de l’Institut Henri Poincar\'e}, 57(4):1901--1939,
  2021.

\bibitem{CL-2018}
P.~Cardaliaguet and C.~Lehalle.
\newblock Mean field game of controls and an application to trade crowding.
\newblock {\em Mathematics and Financial Economics}, 12(3):335--363, 2018.

\bibitem{Carlin2007}
B.~Carlin, M.~Lobo, and S.~Viswanathan.
\newblock Episodic liquidity crises: Cooperative and predatory trading.
\newblock {\em Journal of Finance}, 62(5):2235--2274, 2007.

\bibitem{CDL-2017}
R.~Carmona, F.~Delarue, and D.~Lacker.
\newblock Mean field games of timing and models for bank runs.
\newblock {\em Applied Mathematics \& Optimization}, 76(1):217--260, 2017.

\bibitem{CL-2015}
R.~Carmona and D.~Lacker.
\newblock A probabilistic weak formulation of mean field games and
  applications.
\newblock {\em Annals of Applied Probability}, 25(3):1189--1231, 2015.

\bibitem{C-Jai-2018}
P.~Casgrain and S.~Jaimungal.
\newblock Mean field games with partial information for algorithmic trading.
\newblock {\em arXiv:1803.04094}, 2018.

\bibitem{C-Jai-2018b}
P.~Casgrain and S.~Jaimungal.
\newblock Mean-field games with differing beliefs for algorithmic trading.
\newblock {\em Mathematical Finance}, 30(3):995--1034, 2020.

\bibitem{Drapeau2019}
S.~Drapeau, P.~Luo, A.~Schied, and D.~Xiong.
\newblock An {FBSDE} approach to market impact games with stochastic
  parameters.
\newblock {\em Probability, Uncertainty and Quantitative Risk}, 6(3):237--260,
  2019.

\bibitem{DuTa-2021}
R.~Dumitrescu, M.~Leutscher, and P.~Tankov.
\newblock Control and optimal stopping mean field games: a linear programming
  approach.
\newblock {\em Electronic Journal of Probability}, 26:1--49, 2021.

\bibitem{FruthSchoenebornUrusov14}
A.~Fruth, T.~Sch\"{o}neborn, and M.~Urusov.
\newblock Optimal trade execution and price manipulation in order books with
  time-varying liquidity.
\newblock {\em Mathematical Finance}, 24(4):651--695, 2014.

\bibitem{FGHP-2018}
G.~Fu, P.~Graewe, U.~Horst, and A.~Popier.
\newblock A mean field game of optimal portfolio liquidation.
\newblock {\em Mathematics of Operations Research}, 46(4):1251--1281, 2021.

\bibitem{FHH-2023}
G.~Fu, P.~Hager, and U.~Horst.
\newblock Mean-field liquidation games with market drop-out.
\newblock {\em to appear in Mathematical Finance}, 2023.

\bibitem{FH-2020}
G.~Fu and U.~Horst.
\newblock Mean-field leader-follower games with terminal state constraint.
\newblock {\em SIAM Journal on Control and Optimization}, 58(4):2078--2113,
  2020.

\bibitem{FHX2}
G.~Fu, U.~Horst, and X.~Xia.
\newblock A mean-field control problem of optimal portfolio liquidation with
  semimartingale strategies.
\newblock {\em to appear in Mathematics of Operations Research}, 2022.

\bibitem{FHX1}
G.~Fu, U.~Horst, and X.~Xia.
\newblock Portfolio liquidation games with self exciting order flow.
\newblock {\em Mathematical Finance}, 30(4):1020--1065, 2022.

\bibitem{GatheralSchied11}
J.~Gatheral and A.~Schied.
\newblock {Optimal trade execution under geometric Brownian motion in the
  Almgren and Chriss framework}.
\newblock {\em International Journal of Theoretical and Applied Finance},
  14(3):353--368, 2011.

\bibitem{GH-2017}
P.~Graewe and U.~Horst.
\newblock Optimal trade exection with instantaneous price impact and stochastic
  resilience.
\newblock {\em SIAM Journal on Control and Optimization}, 55(6):3707--3725,
  2017.

\bibitem{PUR}
P.~Graewe, U.~Horst, and R.~Sircar.
\newblock A maximum principle approach to a deterministic mean field game of
  control with absorption.
\newblock {\em SIAM Journal on Control and Optimization}, 60(5):3173--3190,
  2022.

\bibitem{H-2005}
U.~Horst.
\newblock Stationary equilibria in discounted stochastic games with weakly
  interacting players.
\newblock {\em Games and Economic Behavior}, 51(1):83--108, 2005.

\bibitem{HS-2006}
U.~Horst and J.~Scheinkman.
\newblock Equilibria in systems of social interactions.
\newblock {\em Journal of Economic Theory}, 130(1):44--77, 2006.

\bibitem{HXZ-2020}
U.~{Horst}, X.~{Xia}, and C.~{Zhou}.
\newblock {Portfolio liquidation under factor uncertainty}.
\newblock {\em Annals of Applied Probability}, 32(1):80--123, 2022.

\bibitem{HJN-2015}
X.~Huang, S.~Jaimungal, and M.~Nourian.
\newblock Mean-field game strategies for optimal execution.
\newblock {\em Applied Mathematical Finance}, 26:153--185, 2019.

\bibitem{Kratz14}
P.~Kratz.
\newblock An explicit solution of a nonlinear-quadratic constrained stochastic
  control problem with jumps: Optimal liquidation in dark pools with adverse
  selection.
\newblock {\em Mathematics of Operations Research}, 39(4):1198--1220, 2014.

\bibitem{KP-2016}
T.~Kruse and A.~Popier.
\newblock Minimal supersolutions for {BSDE}s with singular terminal condition
  and application to optimal position targeting.
\newblock {\em Stochastic Processes and their Applications}, 126(9):2554--2592,
  2016.

\bibitem{Schied-2020}
X.~Luo and A.~Schied.
\newblock Nash equilibrium for risk-averse investors in a market impact game
  with transient price impact.
\newblock {\em Market Microstructure and Liquidity}, 5, 2020.

\bibitem{Micheli-2023}
A.~Micheli, J.~Muhle-Karbe, and E.~Neuman.
\newblock Closed-loop nash competition for liquidity.
\newblock {\em Mathematical Finance}, 33(4):1082--1118, 2023.

\bibitem{NeumannVoss-2023}
E.~Neuman and M.~Voß.
\newblock Trading with the crowd.
\newblock {\em Mathematical Finance}, 33(3):548--617, 2023.

\bibitem{N-2018}
M.~Nutz.
\newblock A mean field games of optimal stopping.
\newblock {\em SIAM Journal on Control and Optimization}, 56:1206--1221, 2018.

\bibitem{OW-2013}
A.~Obizhaeva and J.~Wang.
\newblock Optimal trading strategy and supply/demand dynamics.
\newblock {\em Journal of Financial Markets}, 16(1):1--32, 2013.

\bibitem{PZ-2018}
A.~Popier and C.~Zhou.
\newblock Second order {BSDE} under monotonicity condition and liquidation
  problem under uncertainty.
\newblock {\em Annals of Applied Probability}, 29(3), 2019.

\bibitem{Schied-2017b}
A.~Schied, E.~Strehle, and T.~Zhang.
\newblock High-frequency limit of {N}ash equilibria in a market impact game
  with transient price impact.
\newblock {\em SIAM Journal on Financial Mathematics}, 8(1):589--634, 2017.

\bibitem{Schied-2019}
A.~Schied and T.~Zhang.
\newblock A market impact game under transient price impact.
\newblock {\em Mathematics of Operations Research}, 44(1):102--121, 2019.

\bibitem{Strehle-2018}
E.~Strehle.
\newblock Optimal execution in a multiplayer model of transient price impact.
\newblock {\em Market Microstructure and Liquidity}, 3(3-4):1850007, 2018.

\bibitem{Teschl-2016}
G.~Teschl.
\newblock {\em Ordinary Differential Equations and Dynamical Systems.}
\newblock AMS, 2016.

\end{thebibliography}

    \appendix

  \section{Proof of Proposition~\ref{prop:sufficient_conditions}}\label{app:A}
    Let us first note that the function $\widetilde{A} := e^{-\int_0^\cdot \frac{\delta\kappa}{\eta_r}dr}(A^{\delta} - \delta \kappa)$ satisfies the Riccati equation $$-\dot{\widetilde{A}} = -\frac{\widetilde{A}^2}{\widetilde\eta} + \widetilde\lambda_t, \quad t \in [0,T) \qquad \lim_{t\to T}\widetilde{A}_t = \infty,$$
    with $\widetilde\eta = \eta e^{-\int_0^\cdot \frac{\delta\kappa}{\eta_r}dr}$ and $\widetilde\lambda = \lambda e^{-\int_0^\cdot \frac{\delta\kappa}{\eta_r}dr}$.
    Hence, by \cite[Lemma A.1]{FHH-2023} we see that
    \begin{align}
        \label{eq:lower_bound_A_delta}
        A^\delta_t - \delta \kappa ~\ge~ e^{\int_0^t \frac{\delta\kappa}{\eta_r}dr}A^{\circ}_t, \qquad t \in [0,T),
    \end{align}
    where $A^{\circ} = (\int_{\cdot}^T \frac{1}{\widetilde{\eta}_s} ds)^{-1}$.
    This implies that $A^\delta - \delta \kappa$ is positive and bounded away from zero.
    Using furthermore \cite[Lemma~A.4]{FHH-2023} we conclude that the function $\alpha^\delta$ is positive, bounded and differentiable.

    Next, we prove that condition (i) and (ii) are sufficient for $\psi^{\delta, \tau}_\mu$ to be strictly decreasing.
    Differentiation of $(\alpha^\delta)^{-1}$ yields that
    \begin{align*}
        \left(\frac{1}{\alpha^\delta_t}\right)^{\prime}
        &= e^{\int_0^t  \frac{A^\delta_r}{\eta_r}\,dr  }\left(\frac{-\frac{(A^\delta_t)^2 }{\eta_t} + \delta\frac{\kappa}{\eta_t} + \lambda_t}{(A^{\delta}_t - \delta\kappa)^2} +\frac{\frac{A^\delta_t}{\eta_t}}{A^\delta_t  - \delta \kappa}\right) ~=~\frac{\lambda_t}{(A^\delta_t - \delta \kappa)\alpha^\delta_t}.
    \end{align*}
    Hence, for the derivative of $\psi^{\delta, \tau}_{\mu}$ we obtain
    \begin{equation}
        \label{dotpsi}
        \begin{split}
            \dot\psi^{\delta,\tau}_\mu(t)
            &=~ \frac{\lambda_t}{(A^\delta_t - \delta \kappa_t)\alpha^\delta_t}\int_t^{\tau} e^{ -\int_0^s\frac{A^{\delta}_r}{\eta_r}\,dr    }\kappa\mu_s\,ds - \frac{1}{A^{\delta}_t - \delta\kappa}\kappa \mu_t \\
            &=~ \frac{1}{A^\delta_t- \delta \kappa}(\lambda_t \psi_\mu(t) - \kappa \mu_t).
        \end{split}
    \end{equation}
    Using that the function $\mu\eta$ satisfies Assumption~\ref{ass:mu_sign}.(ii) we estimate
    \begin{align*}
        \dot\psi^{\delta,\tau}_\mu(t)
        &= \frac{1}{A^\delta_t - \delta \kappa}\left(\frac{\lambda_t}{\alpha^\delta_t}\int_t^{\tau} e^{ -\int_0^s\frac{A^{\delta}_r}{\eta_r}\,dr    }\kappa\mu_s\,ds - \kappa \mu_t\right) \\
        &\le \frac{\kappa }{A^\delta_t - \delta\kappa}\eta_t\mu_t\left(\frac{\lambda_t}{\alpha^\delta_t}\int_t^{T} e^{ -\int_0^s\frac{A^\delta_r}{\eta_r}\,dr    }\frac{1}{\eta_s}\,ds - \frac{1}{\eta_t}\right),
    \end{align*}
    for all $t\in[0, \tau]$. It thus suffices to show that the term in the above bracket is strictly negative assuming either (i) or (ii).
    \begin{itemize}
        \item[(i)]

        In this case, we use \eqref{eq:lower_bound_A_delta} to directly estimate
        \begin{align*}
            \frac{1}{\alpha^\delta_t}\int_t^{T} e^{ -\int_0^s\frac{A^\delta_r}{\eta_r}\,dr}\frac{1}{\eta_s}\,ds &\le \frac{e^{-\int_0^t \frac{\delta\kappa}{\eta_r}dr}}{A^\circ_t}\int_t^{T} e^{ -\int_t^s\frac{A^\delta_r}{\eta_r}\,dr}\frac{1}{\eta_s}\,ds \\
            &= \frac{1}{A^\circ_t }\int_t^{T} e^{ -\int_t^s\frac{A^\delta_r - \delta\kappa}{\eta_r}\,dr}\frac{e^{-\int_0^s \frac{\delta\kappa}{\eta_r}dr}}{\eta_s}\,ds
            \\
            &\le \frac{1}{A^\circ_t}\int_t^{T} e^{ -\int_t^s\frac{A^\circ_r}{\widetilde\eta_r}\,dr}\frac{e^{-\int_0^s \frac{\delta\kappa}{\eta_r}dr}}{\eta_s}\,ds \\
            &=  \frac{1}{A^\circ_t} \int_t^{T} \frac{\int_s^T\frac{1}{\widetilde\eta_u}\,du}{\int_t^T\frac{1}{\widetilde\eta_u}\,du}\frac{e^{-\int_0^s \frac{\delta\kappa}{\eta_r}dr}}{\eta_s}\,ds\\
            &\leq   \int_t^{T}\frac{1}{\widetilde\eta_s} \int_s^T\frac{1}{\widetilde\eta_u}\,du\,ds \\
            &\le\frac{1}{2}\left(\int_t^T\frac{1}{\widetilde\eta_u}\,du\right)^2 \\
        \end{align*}
        for all $t\in [0, \tau]$.
        Recall that for $\delta = 0$, it holds $\widetilde{\eta} = \eta$, and we conclude in this case for $\Vert \lambda \Vert_\infty < \frac{2}{T^2 \Vert \eta \Vert_\infty \Vert \eta^{-1} \Vert_\infty^2 }$. 
        In the case $\delta > 0$ we can continue the estimate to obtain a more explicit bound:
        \begin{align*}
           \frac{1}{2}\left(\int_t^T\frac{1}{\widetilde\eta_u}\,du\right)^2 & = \frac{1}{2}\left(\int_t^T\frac{1}{\eta_u} e^{\int_0^u \frac{\delta\kappa}{\eta_r}dr}\,du\right)^2\\
          &= \frac{1}{2} \left(\frac{1}{\delta \kappa}\left(e^{\int_0^T \frac{\delta\kappa}{\eta_r}dr} - e^{\int_0^t \frac{\delta\kappa}{\eta_r}dr}\right)\right)^2\\
        &\le \frac{1}{2}\left(\frac{1}{\delta \kappa}\left(e^{\int_0^T \frac{\delta\kappa}{\eta_r}dr} - 1\right)\right)^2\\
        &\le \frac{1}{2}\left(e^{\delta\kappa\int_0^T \frac{1}{\eta_r}dr}\int_0^T \frac{1}{\eta_r}dr\right)^2\\
        &\le \frac{T^2}{2} \Vert \eta^{-1}\Vert_{\infty}^2 e^{ 2T \delta\kappa \Vert \eta^{-1}\Vert_{\infty}},
        \end{align*}
        where in the second last inequality we have used that $\frac{1}{x}(e^{ax} - 1) \le a e^{a x}$ for all $a,x>0$.
        We readily conclude that $\Vert \lambda \Vert_{\infty} <    2 (T^2\Vert \eta^{-1}\Vert_{\infty}^2 \|\eta\|_\infty)^{-1} e^{- 2T \kappa \Vert \eta^{-1}\Vert_{\infty}} $ is sufficient since $\delta\leq 1$.
        \item[(ii)]In this case we consider the function
        \begin{align*}
            z(t) := \frac{1}{\alpha^\delta_t}\int_t^{T} e^{ -\int_0^s\frac{A^\delta_r}{\eta_r}\,dr}\frac{1}{\eta_s}\,ds, \qquad 0 \le t \le T.
        \end{align*}
        Differentiation yields for all $0 \le t < T$ that
        \begin{align*}
            \dot{z}(t) =&~ \frac{\lambda_t e^{ \int_0^t\frac{A^\delta_r}{\eta_r}\,dr}}{(A_t-\delta\kappa)^2}\int_t^{T} e^{ -\int_0^s\frac{A^\delta_r}{\eta_r}\,dr}\frac{1}{\eta_s}\,ds - \frac{1}{\alpha^\delta_t} e^{ -\int_0^t\frac{A^\delta_r}{\eta_r}\,dr} \frac{1}{\eta_t}\\
            =&~ \frac{\lambda_t }{(A^\delta_t -\delta\kappa) \alpha^\delta_t}\int_t^{T} e^{ -\int_0^s\frac{A^\delta_r}{\eta_r}\,dr}\frac{1}{\eta_s}\,ds
            - \frac{1}{A^\delta_t-\delta\kappa}\frac{1}{\eta_t} \\
            =&~ \frac{1}{(A^\delta_t - \delta\kappa) \eta_t }\left( \lambda_t \eta_t z(t) - 1\right).
        \end{align*}
        If we can prove that
        \[
            t_0 := \sup\left\{ t \in [0,T] \;\vert\; z(t)\lambda_t\eta_t = 1\right\} = - \infty,
        \]
        then we have for all $t\in[0,T]$ that
        \begin{align*}
            \frac{\lambda_t}{\alpha^\delta_t}\int_t^{T} e^{ -\int_0^s\frac{A^\delta_r}{\eta_r}\,dr    }\frac{1}{\eta_s}\,ds  = z(t) \lambda_t < \frac{1}{\eta_t}
        \end{align*}
        and the claim readily follows.
        Since, $z(T) = 0$ it holds that $t_0 < T$. Let us now assume to the contrary that $t_0 \ge 0$. Then $\lambda_{t_0} > 0$, $z(t_0) = \frac{1}{\lambda_{t_0}\eta_{t_0}}$, $\dot{z}(t_0) = 0$ and
        \[
            z(t) < \frac{1}{\lambda_t \eta_t } \quad \mbox{for all} \quad t\in ({t_0},T].
        \]
        Since $(\lambda \eta)$ is non-decreasing $\dot{z}(t) \ge \frac{1}{(A^\delta_t - \delta\kappa) \eta_t}(\lambda_{t_0}\eta_{t_0} z(t) - 1)$ on $[t_0, T)$ and hence
        we have for all $t \in [t_0, T)$
        \[
            \left(z(t) - \frac{1}{\lambda_{t_0}\eta_{t_0}}\right)^{\prime} \ge \frac{\lambda_{t_0}\eta_{t_0}}{(A^\delta_t-\delta\kappa)\eta_t}\left(z(t) - \frac{1}{\lambda_{t_0}\eta_{t_0}}\right).
        \]
        By Gr\"onwall's inequality this shows that
        \[
            z(t) - \frac{1}{\lambda_{t_0}\eta_{t_0}} \ge  \left( z(t_0) - \frac{1}{\lambda_{t_0}\eta_{t_0}}\right)e^{\int_{t_0}^t \frac{\lambda_{t_0}\eta_{t_0}}{(A^\delta_s-\delta\kappa)\eta_s }\,ds} = 0,
        \]
        and hence $z(t)\ge \frac{1}{\lambda_{t_0}\eta_{t_0}}  \ge \frac{1}{\lambda_t\eta_t}$ on $[t_0, T)$, which contradicts the definition of $t_0$.
    \end{itemize}

\section{Verification result for buyers}

\subsection{Proof of Lemma \ref{lem:admiss}}

        By Lemma \ref{lem-admis} the strategy is square integrable and satisfies the liquidation constraint. Furthermore, it follows from the definition of $\sigma_\mu(x)$ and the representation \eqref{Yst} of the adjoint process that
        \[
            Y^{\delta, \sigma_\mu(x),T}_{\sigma_\mu(x)}  - \delta\kappa X^{\delta,\sigma_\mu(x),T}_{\sigma_\mu(x)}  < 0 \quad \mbox{if} \quad |x| > \|\psi^{\delta,T}_\mu\|_\infty
        \]
        and
        \[
            Y^{\delta,\sigma_\mu(x),T}_{\sigma_\mu(x)} - \delta\kappa X^{\delta,\sigma_\mu(x),T}_{\sigma_\mu(x)} = 0  \quad \mbox{if} \quad |x| \leq \|\psi^{\delta,T}_\mu\|_\infty.
        \]

        Furthermore, for every $t_0 \in [\sigma^\tau_\mu(x), T]$ with $Y^{\delta, \sigma_\mu(x),T}_{t_0}  -\delta\kappa X^{\delta,\sigma_\mu(x),T}_{t_0}= 0$ it follows from the equation \eqref{Yst} that
        \begin{align*}
            X^{\delta, \sigma_\mu(x),T}_{t_0} = -\psi^{\delta,T}_\mu(t_0),
        \end{align*}
        and hence from the ODE for $Y^{\delta, \sigma_\mu(x),T}$ and the definition of $\psi^{\delta,T}_\mu$ that
        \begin{equation*}
            \begin{split}
                \big( Y^{\delta, \sigma_\mu(x),T}_{t_0} -  \delta\kappa X^{\delta,\sigma_\mu(x),T}_{t_0} \big)' &= \lambda_{t_0}\psi^{\delta,T}_\mu(t_0) - \mu_{t_0}\kappa \\
                &= (A^\delta_{t_0} -\delta\kappa)  \dot\psi^{\delta,T}_\mu(t_0) \\
                &< 0.
            \end{split}
        \end{equation*}
        In particular, the process $Y^{\delta, \sigma_\mu(x),T}-\delta\kappa X^{\delta,\sigma_\mu(x),T}$ is strictly negative in a vicinity of the entry time and strictly decreasing in a vicinity of every time it hits zero. As a result,
        \[
            Y^{\delta, \sigma_\mu(x),T} -\delta\kappa X^{\delta, \sigma_\mu(x), T} < 0 \quad \mbox{on} \quad (\sigma_\mu(x),T].
        \]
        Hence the strategy is admissible in a model with trading constraints.

\subsection{Proof of Theorem \ref{thm:veri-N}}

        Let $\xi$ be an arbitrary strategy of player $i$ with corresponding portfolio process $X$ and market entry time $\sigma$. We distinguish two cases, depending on which strategy enters the market first.\footnote{To unify the notion for finite player and MFGs, the case $N=\infty$ corresponds to the MFG. In this case many terms drop out and the computation simplifies.}

        \begin{itemize}
            \item Let $\sigma < \sigma^{*,i}$. In particular $\sigma^{*,i}>0$. To compare the transaction costs $J_i(\xi;\xi^{-i})$ and $J_i(\xi^{*,i};\xi^{-i})$, we split the cost functions into three terms as follows:
            \begin{equation*}
                \begin{split}
                    J_i(\xi;\xi^{-i}) = &~\int_0^\sigma \left\{\kappa\left(\mu_s+\frac{1}{N}\xi_s - \frac{1}{N}\xi^{*,i}_s\right)X_s+\frac{1}{2}\lambda_s X^2_s+\frac{1}{2}\eta_s\xi^2_s\right\} \,ds\\
                    &~ +  \int_\sigma^{\sigma^{*,i}}\left\{\kappa\left(\mu_s+\frac{1}{N}\xi_s - \frac{1}{N}\xi^{*,i}_s\right)X_s+\frac{1}{2}\lambda_s X^2_s+\frac{1}{2}\eta_s\xi^2_s\right\} \,ds\\
                    &~+ \int_{\sigma^{*,i}}^T \left\{\kappa\left(\mu_s+\frac{1}{N}\xi_s - \frac{1}{N}\xi^{*,i}_s\right)X_s+\frac{1}{2}\lambda_s X^2_s+\frac{1}{2}\eta_s\xi^2_s\right\} \,ds\\
                    =&~\int_0^\sigma \left\{\kappa \mu_s x_i+\frac{1}{2}\lambda_s x^2_i \right\} \,ds\\
                    &~ +  \int_\sigma^{\sigma^{*,i}}\left\{\kappa\left(\mu_s+\frac{1}{N}\xi_s  \right)X_s+\frac{1}{2}\lambda_s X^2_s+\frac{1}{2}\eta_s\xi^2_s\right\} \,ds\\
                    &~+ \int_{\sigma^{*,i}}^T \left\{\kappa\left(\mu_s+\frac{1}{N}\xi_s - \frac{1}{N}\xi^{*,i}_s\right)X_s+\frac{1}{2}\lambda_s X^2_s+\frac{1}{2}\eta_s\xi^2_s\right\} \,ds
                \end{split}
            \end{equation*}
            and
            \begin{equation*}
                \begin{split}
                    J_i(\xi^{*,i};\xi^{-i}) = &~\int_0^\sigma \left\{\kappa \mu_s x_i+\frac{1}{2}\lambda_s x_i^2 \right\} \,ds\\
                    &~ +  \int_\sigma^{\sigma^{*,i}}\left\{\kappa \mu_s x_i+\frac{1}{2}\lambda_s x_i^2 \right\} \,ds \\
                    &~+ \int_{\sigma^{*,i}}^T  \left\{\kappa \mu_s X^{*,i}_s+\frac{1}{2}\lambda_s (X_s^{*,i})^2+\frac{1}{2}\eta_s(\xi_s^{*,i})^2\right\} \,ds.
                \end{split}
            \end{equation*}

            Thus, using convexity in the second step, we obtain that
            \begin{equation*}
                \begin{split}
                    &~J_i(\xi;\xi^{-i}) - J_i(\xi^{*,i};\xi^{-i})\\
                    =&~ \int_\sigma^{\sigma^{*,i}}\left\{ \kappa  \mu_sX_s+\frac{1}{2}\lambda_s X^2_s -\kappa \mu_s x_i - \frac{1}{2}\lambda_s x_i^2\right\} \,ds + \int_\sigma^{\sigma^{*,i}}  \frac{\kappa}{N}\xi_s X_s\,ds +\frac{1}{2}\int_\sigma^{\sigma^{*,i}}\eta_s\xi^2_s\,ds  \\
                    &~ +\int_{\sigma^{*,i}}^T \left\{\kappa \mu_s (X_s-X^{*,i}_s)+\frac{1}{2}\lambda_s X^2_s+\frac{1}{2}\eta_s\xi^2_s-\frac{1}{2}\lambda_s (X_s^{*,i})^2 - \frac{1}{2}\eta_s(\xi_s^{*,i})^2\right\} \,ds \\
                    &~ +\int_{\sigma^{*,i}}^T \frac{ \kappa}{N}\left( \xi_s - \xi^{*,i}_s\right)X_s\,ds \\
                    \geq&~\int_\sigma^{\sigma^{*,i}}\left\{ \kappa  \mu_sX_s+\frac{1}{2}\lambda_s X^2_s -\kappa \mu_s x_i - \frac{1}{2}\lambda_s x_i^2\right\} \,ds + \int_\sigma^{\sigma^{*,i}}  \frac{\kappa}{N}\xi_s X_s\,ds +\frac{1}{2}\int_\sigma^{\sigma^{*,i}}\eta_s\xi^2_s\,ds  \\
                    &~ +\int_{\sigma^{*,i}}^T \left\{\kappa \mu_s (X_s-X^{*,i}_s) + \lambda_s (X_s-X^{*,i}_s)X^{*,i}_s+ \eta_s(\xi_s-\xi^{*,i}_s)\xi^{*,i}_s \right\} \,ds \\
                    & ~ +\int_{\sigma^{*,i}}^T \frac{ \kappa}{N}\left( \xi_s - \xi^{*,i}_s\right)X_s\,ds.
                \end{split}
            \end{equation*}

            Due to the constant market impact $\kappa$ and since $-\xi = \dot X$ and $-\xi^{*,i} = \dot X^{*,i}$, the last term on the right hand side of the above inequality satisfies
            \begin{equation*}
                \begin{split}
                    &~\int_{\sigma^{*,i}}^T \frac{ \kappa}{N}\left( \xi_s - \xi^{*,i}_s\right)X_s\,ds\\
                    =&~ \int_{\sigma^{*,i}}^T\frac{ \kappa}{N} \left( \xi_s - \xi^{*,i}_s\right)(X_s-X^{*,i}_s)\,ds+ \int_{\sigma^{*,i}}^T \frac{ \kappa}{N}\left( \xi_s - \xi^{*,i}_s\right)X^{*,i}_s\,ds\\
                    =&~\frac{\kappa}{2N}(X_{\sigma^{*,i}}-X^{*,i}_{\sigma^{*,i}})^2+\int_{\sigma^{*,i}}^T\frac{ \kappa}{N} \left( \xi_s - \xi^{*,i}_s\right)X^{*,i}_s\,ds.
                \end{split}
            \end{equation*}

            To simplify the second to last term, we recall that the strictly positive entry time $\sigma^{*,i}$ satisfies
            \[
                \frac{\kappa}{N}x_i = Y_{\sigma^{*,i}}^i.
            \]
            Hence, integration by parts yields that
            \begin{equation*}
                \begin{split}
                    \frac{\kappa}{N}x_i ( X_{\sigma^{*,i}}-X^{*,i}_{\sigma^{*,i}} ) & ~ = Y_{\sigma^{*,i}}^i( X_{\sigma^{*,i}}-X^{*,i}_{\sigma^{*,i}} ) \\
                    & ~ = \int_{\sigma^{*,i}}^T Y^i_s( \xi_s-\xi^{*,i}_s )\,ds +\int_{\sigma^{*,i}}^T (X_s-X^{*,i}_s)(\lambda_sX^{*,i}_s+\kappa\mu_s)\,ds,
                \end{split}
            \end{equation*}
            and so the second to last term equals
            \[
                \int_{\sigma^{*,i}}^T \left\{  -Y^i_s( \xi_s-\xi^{*,i}_s )+ \eta_s(\xi_s-\xi^{*,i}_s)\xi^{*,i}_s \right\} \,ds + \frac{\kappa}{N}x_i ( X_{\sigma^{*,i}}-X^{*,i}_{\sigma^{*,i}} ).
            \]
            This shows that
            \begin{equation*}
                \begin{split}
                    &~J_i(\xi;\xi^{-i}) - J_i(\xi^{*,i};\xi^{-i})\\
                    \geq&~\int_\sigma^{\sigma^{*,i}}\left\{ \kappa  \mu_sX_s+\frac{1}{2}\lambda_s X^2_s -\kappa \mu_s x_i - \frac{1}{2}\lambda_s x_i^2\right\} \,ds + \int_\sigma^{\sigma^{*,i}}  \frac{\kappa}{N}\xi_s X_s\,ds +\frac{1}{2}\int_\sigma^{\sigma^{*,i}}\eta_s\xi^2_s\,ds  \\
                    & + \int_{\sigma^{i,*}}^T \left\{ -Y^i_s( \xi_s-\xi^{*,i}_s )+ \eta_s(\xi_s-\xi^{*,i}_s)\xi^{*,i}_s + \frac{\kappa}{N} \left( \xi_s - \xi^{*,i}_s\right)X^{*,i}_s \right\}\,ds \\
                    &~ +\frac{\kappa}{2N}(X_{\sigma^{*,i}}-X^{*,i}_{\sigma^{*,i}})^2 + \frac{\kappa}{N}x_i ( X_{\sigma^{*,i}}-X^{*,i}_{\sigma^{*,i}} ).
                \end{split}
            \end{equation*}

            Using the fact that $\xi^{*,i} = \frac{Y^{i} -\frac 1 N \kappa X^{*,i}}{\eta}$ on $[\sigma^{*,i}, T]$ we see that the third line above vanishes and so
            \begin{equation*}
                \begin{split}
                    &~J_i(\xi;\xi^{-i}) - J_i(\xi^{*,i};\xi^{-i}) \\
                    \geq &~\int_\sigma^{\sigma^{*,i}}\left\{ \kappa  \mu_sX_s+\frac{1}{2}\lambda_s X^2_s -\kappa \mu_s x_i - \frac{1}{2}\lambda_s x_i^2\right\} \,ds + \int_\sigma^{\sigma^{*,i}}  \frac{\kappa}{N}\xi_s X_s\,ds +\frac{1}{2}\int_\sigma^{\sigma^{*,i}}\eta_s\xi^2_s\,ds \\
                    & ~ +   \frac{\kappa}{2N}(X_{\sigma^{*,i}}-X^{*,i}_{\sigma^{*,i}})^2 +  \frac{\kappa}{N}x_i( X_{\sigma^{*,i}}-X^{*,i}_{\sigma^{*,i}} )   \\
                    =&~\int_{\sigma}^{ \sigma^{*,i} }\left\{\kappa\mu_sX_s+\frac{1}{2}\lambda_s X^2_s-\kappa \mu_s x_i- \frac{1}{2}\lambda_s x^2_i\right\}\,ds + \int_\sigma^{\sigma^{*,i}}\frac{1}{2}\eta_s \xi^2_s\,ds \\
                    & ~ +   \frac{\kappa}{2N}(X_{\sigma^{*,i}}-X^{*,i}_{\sigma^{*,i}})^2 +  \frac{\kappa}{N}x_i( X_{\sigma^{*,i}}-X^{*,i}_{\sigma^{*,i}} )  -  \frac{\kappa}{2N}( X^2_{\sigma^{*,i}} - X^2_{\sigma})    \\
                    =&~\int_{\sigma}^{ \sigma^{*,i} }\left\{(\kappa\mu_s+\lambda_s x_i)( X_s-x_i )+\frac{1}{2}\lambda_s (X_s-x_i)^2 \right\} \,ds + \int_\sigma^{\sigma^{*,i}}\frac{1}{2}\eta_s\xi^2_s\,ds \\
                    & ~ +  \frac{\kappa}{2N}(X_{\sigma^{*,i}}-x_i)^2 +  \frac{\kappa}{N}x_i( X_{\sigma^{*,i}}-x_i ) - \frac{\kappa}{2N}( X^2_{\sigma^{*,i}} -x_i^2 )   \\
                    =&~  \int_{\sigma}^{ \sigma^{*,i} }\left\{(\kappa\mu_s+\lambda_s x_i)( X_s-x_i )+\frac{1}{2}\lambda_s (X_s-x_i)^2 \right\} \,ds + \int_\sigma^{\sigma^{*,i}}\frac{1}{2}\eta_s\xi^2_s\,ds .
                \end{split}
            \end{equation*}

            Since the process $\mu$ satisfies Assumption~\ref{ass:mu_sign} it follows from Assumption \ref{ass:ceof_relation} that
            \[
                \dot\psi_{\mu}(s)=\frac{1}{A^\delta_s-\frac{\kappa}{N}}\left(\lambda_s\psi_{\mu}(s)-\kappa\mu_s\right)<0 \quad \mbox{on} \quad (0,T],
            \]
where $\psi_\mu := \psi_{\mu}^{\delta,T}$ is defined in \eqref{eq:psi} with $\tau=T$. 
            
            Since $\sigma^{*,i} > 0$ we also have that $-\psi_{\mu}(\sigma^{*,i})=x_i$, hence, that
            \[
                \kappa\mu_s+\lambda_s x_i\geq \kappa\mu_s-\lambda_s\psi_{\mu}(s)>0
                \quad \mbox{and} \quad
                X - x_i \geq 0
                \quad \mbox{on} \quad [\sigma, \sigma^{*,i}].
            \]
            Thus we conclude that
            \begin{align*}
                J_i(\xi;\xi^{-i}) - J_i(\xi^{*,i};\xi^{-i})  ~\geq ~\int_\sigma^{\sigma^{*,i}} \left\{\frac{1}{2}\lambda_s (X_s-x_i)^2 +  \frac{1}{2}\eta_s\xi^2_s \right\} \,ds ~\geq ~ 0.
            \end{align*}

            \item The case $\sigma \geq \sigma^{*,i}$ is simpler. In this case,
            \begin{align*}
                &~ J_i(\xi;\xi^{-i}) - J_i(\xi^{*,i};\xi^{-i})\\
                = &~\int_{ \sigma^{*,i} }^T\left\{ \frac{1}{2} \eta_s \xi_s^2 + \kappa\mu_s X_s + \frac{1}{2}\lambda_sX_s^2 +\frac{\kappa}{N}X_s(\xi_s-\xi^{*,i}_s) \right\}  \,ds \\
                & ~ - \int_{ \sigma^{*,i} }^T \left\{ \frac{1}{2} \eta_s (\xi^{*,i}_s)^2 + \kappa\mu_s X^{*,i}_s + \frac{1}{2}\lambda_s(X^{*,i}_s)^2  \right\}  \,ds  \\
                \geq&~  \int_{ \sigma^{*,i} }^T \left\{ \eta_s \xi^{*,i}_s( \xi_s-\xi^{*,i}_s ) + \kappa\mu_s( X_s-X^{*,i}_s ) + \lambda_s X^{*,i}_s( X_s-X^{*,i}_s )+\frac{\kappa}{N}X^{*,i}_s(\xi_s-\xi^{*,i}_s) \right\}  \,ds\\
                &~  + \int_{\sigma^{*,i}}^T     \frac{\kappa}{N}(X_s-X^{*,i}_s)(\xi_s-\xi^{*,i}_s)\,ds .  
            \end{align*}
            First, $ \int_{\sigma^{*,i}}^T     \frac{\kappa}{N}(X_s-X^{*,i}_s)(\xi_s-\xi^{*,i}_s)\,ds = -\left.\frac{\kappa}{2N}( X_s-X^{*,i}_s )^2\right|_{\sigma^{*,i}}^T =0 $.
            Second, applying integration by parts to $Y^i(X-X^{*,i})$ on $[\sigma^{*,i},T]$ and noting that $X_{\sigma^{*,i}}=X^{*,i}_{\sigma^{*,i}}=x_i$, we have that
            \[
                0=Y^i_{\sigma^{*,i}}(X_{\sigma^{*,i}}-X^{*,i}_{\sigma^{*,i}}) =\int_{\sigma^{*,i}}^T Y^i_s(\xi_s-\xi^{*,i}_s)\,ds + \int_{\sigma^{*,i}}^T(X_s-X^{*,i}_s)( \lambda_s X^{*,i}_s+\kappa\mu_s )\,ds,
            \]
            which implies that
            \begin{align*}
                J_i(\xi;\xi^{-i}) - J_i(\xi^{*,i};\xi^{-i})
                \geq  \int_{ \sigma^{*,i} }^T  \left(\eta_s \xi^{*,i}_s + \frac{\kappa}{N}X^{*,i}_s -Y^i_s  \right) ( \xi_s-\xi^{*,i}_s )        \,ds
                = 0.
            \end{align*}
        \end{itemize}
        Now assume $\xi$ is another optimal strategy. The above argument leads to $0\geq J_i(\xi;\xi^{-i})-J_i(\xi^{*,i};\xi^{-i})\geq 0$. Thus, all above inequalities become equalities. As a result, $\xi=\xi^{*,i}$ in both cases.

\end{document}